\documentclass[12pt,a4paper]{article}
\usepackage[utf8]{inputenc}
\usepackage{lmodern}
\usepackage[T1]{fontenc}
\usepackage[english]{babel}
\usepackage{ifpdf}
\usepackage[usenames,dvipsnames]{xcolor}
\usepackage{graphicx}
\usepackage[shortlabels]{enumitem}
\usepackage[numbers,sort]{natbib}
\usepackage[font=footnotesize,width=\textwidth]{caption}
\usepackage{sectsty}
\usepackage{hyperref}
\usepackage{uri}
\usepackage[top=1.8cm, bottom=2.8cm]{geometry}
\usepackage[affil-it]{authblk}
\usepackage[parfill]{parskip}
\usepackage[setbuildercolon]{matmatix}

\allowdisplaybreaks

\ifpdf
  \hypersetup{
    pdftitle={The B2 index of galled trees}
  }
\fi

\hypersetup{colorlinks, allcolors = {MidnightBlue}}

\allsectionsfont{\boldmath}

\newenvironment{mathlist}
{\begin{enumerate}[label={\upshape(\roman*)}, align=left, widest=iii, leftmargin=*]}
{\end{enumerate}\ignorespacesafterend}

\makeatletter
\renewcommand*{\defas}{%
\mathrel{\rlap{\raisebox{0.3ex}{$\m@th\cdot$}}\raisebox{-0.3ex}{$\m@th\cdot$}}=}
\newcommand*{\asdef}{%
\mathrel{{=\llap{\raisebox{0.3ex}{$\m@th\cdot$}}\llap{\raisebox{-0.3ex}{$\m@th\cdot$}}}}}
\makeatother

\makeatletter
\newtheorem*{rep@theorem}{\rep@title}
\newcommand{\newreptheorem}[2]{%
\newenvironment{rep#1}[1]{%
 \def\rep@title{#2 \ref*{##1}}%
 \begin{rep@theorem}}%
 {\end{rep@theorem}}}
\makeatother

\newtheorem{theorem}{Theorem}[section]
\newreptheorem{theorem}{Theorem}
\newtheorem{proposition}[theorem]{Proposition}
\newreptheorem{proposition}{Proposition}
\newtheorem{lemma}[theorem]{Lemma}
\newreptheorem{lemma}{Lemma}
\newtheorem{corollary}[theorem]{Corollary}
\theoremstyle{definition}
\newtheorem{definitionX}[theorem]{Definition}
\newenvironment{definition}
  {\pushQED{\qed}\definitionX}
  {\popQED\enddefinitionX}

\newtheorem{remarkX}[theorem]{Remark}
\newenvironment{remark}
  {\pushQED{\qed}\remarkX}
  {\popQED\endremarkX}

\newcommand{\ancestor}{\preccurlyeq}
\newcommand{\finer}{\preccurlyeq}

\newcommand{\codir}{\mathrel{\:\,\mathclap{\downdownarrows} \mathclap{\raisebox{0.2ex}{$\equiv$}}\:\,}}
\newcommand{\stub}{\sqsubset}
\newcommand{\phylnet}{\mathds{G}}
\newcommand{\Head}[1]{\Gamma_{\!#1}}

\DeclareMathOperator{\limup}{lim\!\uparrow\,}

\title{\texorpdfstring{\vspace{-1ex}}{}The $B_2$ index of galled trees}

\begin{document}

\author[1,2]{François Bienvenu}
\author[1]{Jean-Jil Duchamps}
\author[3]{\break Michael Fuchs}
\author[4]{Tsan-Cheng Yu}

\renewcommand\Affilfont{\itshape\small\renewcommand{\baselinestretch}{1}}

\affil[1]{Univ.\ Marie et Louis Pasteur, CNRS, LmB (UMR 6623), F-25000 Besançon, France.\vspace{1ex}}
\affil[2]{Institute for Theoretical Studies, ETH Zürich,
8092 Zürich, Switzerland.\vspace{1ex}}
\affil[3]{Department of Mathematical Sciences, National Chengchi University,\linebreak
Taipei 116, Taiwan.}
\affil[4]{Department of Mathematics, Fu Jen Catholic University,\linebreak
New Taipei City 242, Taiwan.}

\maketitle

\begin{abstract}
In recent years, there has been an effort to extend the classical notion of
phylogenetic balance, originally defined in the context of trees, to networks.
One of the most natural ways to do this is with the so-called $B_2$
index. In this paper, we study the $B_2$~index for a prominent class of
phylogenetic networks: galled trees.  We show that the $B_2$ index of a
uniform leaf-labeled galled tree converges in distribution as the network
becomes large. We characterize the corresponding limiting distribution, and provide
a way to compute its moments. This is the first time that a balance
index has been studied to this level of detail for a random
phylogenetic network.

One specificity of this work is that we use two different and independent
approaches, each with its advantages: analytic combinatorics, and local
limits. The analytic combinatorics approach is more direct, as it relies on
standard tools; but it involves slightly more complex calculations. Because
it has not previously been used to study such questions, the local limit
approach requires developing an extensive framework beforehand; however, this
framework is interesting in itself and can be used to tackle other similar
problems.
\end{abstract}

{\small
}

\setcounter{tocdepth}{2}
\tableofcontents

\section{Introduction} \label{secIntro}

\subsection{Biological context and main result} \label{secBiology}

Indices of phylogenetic balance -- \emph{balance indices} for short --
are summary statistics that quantify the idea that some trees have more
symmetries than others. They play a central role in theoretical evolutionary
biology, and are also widely used in practical applications.  There is a great
diversity of balance indices, and their mathematical properties are for the
most part well-understood; see \cite{fischer2023tree} for a comprehensive
survey.

However, with the growing recognition of the importance of reticulate evolution,
phylogenetic networks are increasingly used to describe phylogenies.
Extending the definition and study of balance indices to networks is thus
becoming an important topic in mathematical phylogenetics, and very
recently some authors have started studying extensions of existing balance
indices in networks \cite{zhang2022sackin} and designing new balance indices
specifically with networks in mind~\cite{knuver2023weighted}.

In that context, one balance index known as the $B_2$ index stands out.
Introduced in the context of trees in \cite{shao1990tree}, this balance index
had gone largely unnoticed. However, as recently pointed out
in~\cite{bienvenu2021revisiting}, where it was studied extensively, its
formal definition and its interpretation as a balance index are unchanged
in the context of networks; we come back to this in Section~\ref{secSetting}
below.

Although it has been studied for various models of random trees,
the distribution of the $B_2$ index has not been studied for random
phylogenetic networks. In this paper, we do so for a natural and well-studied
model: uniform leaf-labeled galled trees (note that there is some variation
regarding the term \emph{galled tree} -- see \cite{cardona2010comparison} for a
detailed discussion -- and that our use of the term follows the prevalent
usage~\cite{huson2010phylogenetic} and refers to what are also known as
\emph{level-1 networks}).

\enlargethispage{3ex}

We show that the $B_2$ index of a galled tree $G_n$ sampled uniformly
at random among galled trees with $n$ labeled leaves
converges in distribution as $n$ goes to infinity, and we give a
characterization of the limiting distribution. We also show the
convergence of all $p$-th moments and that, in particular,
$\ExpecBrackets{B_2(G_n)}$ converges to a constant $c = 2.707911858984\ldots$

This is the first result of this type for the balance index
of a random phylogenetic network: the only comparable result in the
current literature is that the expected value of an extension of the Sackin index of a
uniform simplex network with $n$ labeled leaves is asymptotically
$\Theta(n^{7/4})$, as proved in~\cite{zhang2022sackin}. However,
beyond the fact that this result is less precise than ours, it is not clear
that the extension of the Sackin index considered in that paper is a measure of
phylogenetic balance.
More results concerning the distribution of balance indices in phylogenetic
networks will be published in a forthcoming paper by
one of the authors \cite{fuchs2024sackin}. In this paper,
the limiting distribution (after scaling) of various extensions of the Sackin
index is studied for several models of phylogenetic networks -- including the
galled trees considered here -- and is shown to be the Airy distribution.

A notable feature of this paper is that we use two very different -- and
completely independent -- approaches to study the asymptotic
behavior of $B_2(G_n)$:
\begin{itemize}
  \item a combinatorial approach, based on analytic combinatorics;
  \item a probabilistic approach, based on modern tools from the study of
    branching processes: local limits of conditioned Galton--Watson trees.
\end{itemize}

The point of presenting these two approaches in the same paper, even though it
means proving the main results twice, is that each approach has its advantages and
disadvantages. The combinatorial approach is more
direct, because it relies on a well-established framework (namely, the method of {\it singularity analysis}; see \cite{flajolet2009analytic}). However, it involves
calculations that are very specific to the problem at hand; and while it is possible to characterize the limiting distribution
of $B_2(G_n)$ through its moments, this requires increasingly
strenuous calculations as the moments get higher.
The probabilistic approach, on the other hand, requires establishing a
substantial number of technical preliminaries. This is because, although local
limits have become a standard tool in probability theory, there are
difficulties regarding the definition of $B_2$ for infinite phylogenetic
networks and its continuity for the local topology. This preliminary work is
thus needed to justify the rigorousness of the approach; but the numerical
calculations themselves are then fairly simple, and immediately yield a simple
recursive characterization of the limiting distribution of $B_2(G_n)$. Moreover,
now that the technicalities have been worked out, it is straightforward to
apply this approach to any other random phylogenetic network with a suitable
local limit.

In the rest of this section, we formally define the $B_2$ index
and the class of galled trees. We then split the paper into two independent
sections: Section~\ref{secAnalyticCombinatorics} details the combinatorial
approach, and Section~\ref{secLocalLimit} the probabilistic one.
In order to emphasize the practicality of the latter as a tool to
perform concrete computations, only the main ideas of the framework and its
application to galled trees are presented in the main text; its technical
justification is detailed in the Appendix.

\subsection{Setting and notation} \label{secSetting}

Let us start by specifying what we mean by \emph{phylogenetic network}.
In order to prepare the ground for the local limit approach, we include
infinite graphs in our definition.

\break

Recall that a DAG is a directed graph with no directed cycles.
It is:
\begin{itemize}
  \item \emph{countable} if its vertex set is finite or countably infinite;
  \item \emph{locally finite} if every vertex has finite in-degree and
    out-degree;
  \item \emph{rooted} if there is a unique vertex with in-degree 0, called the
    \emph{root}, and every vertex $v$ is reachable from the root (that is,
    there exists a directed path going from the root to $v$).
\end{itemize}

\begin{definition} \label{DefPhyNet}
A \emph{phylogenetic network} is a countable, locally finite rooted DAG.
Its vertices with out-degree~0, if any, are called the \emph{leaves}.
\end{definition}

Next, we recall the definition of the $B_2$ index of a \emph{finite}
phylogenetic network, as given in~\cite{bienvenu2021revisiting}.
Extending this definition to infinite phylogenetic networks will be one of the
main challenges for the local limit approach.

Let $G$ be a finite phylogenetic network, and let $X = (X_t)_{t\geq 0}$ be the
simple random walk on $G$, started from the root and constrained to follow the
direction of the edges: at each step, one of the outgoing edges of the current
vertex is chosen uniformly at random, until a leaf is reached. In the rest of
this document, we will refer to $X$ simply as \emph{the directed random walk
on~$G$}.  For each vertex $v \in G$, let $p_v$ denote the probability that $X$
visits~$v$. In a finite phylogenetic network, $(p_\ell)_{\ell \in \mathcal{L}}$
is a probability distribution on~$\mathcal{L}$, the leaf set of~$G$.

\begin{definition} \label{defB2Finite}
Let $G$ be a finite phylogenetic network. The \emph{$B_2$ index of $G$} 
is the Shannon entropy of the probability distribution
$(p_\ell)_{\ell \in \mathcal{L}}$ induced on the leaves of $G$ by the directed
random walk -- that is, 
\[
  B_2(G) \;=\;
  - \sum_{\ell \in \mathcal{L}} p_\ell \log_2 p_\ell \, .\qedhere
\]
\end{definition}

We only recall, without proof, an elementary but fundamental property of
the $B_2$ index that will be of constant use in this paper. We refer the reader
to \cite{bienvenu2021revisiting} for a more detailed presentation.

\begin{proposition}[grafting property] \label{propGraftingFinite}
Let $G_1$ and $G_2$ be two finite phylogenetic networks, and let $G$ be the
phylogenetic network obtained by identifying a leaf $\ell \in G_1$ with the
root of $G_2$. Then,
\[
  B_2(G) \;=\; B_2(G_1) \;+\; p_\ell\, B_2(G_2) \,, 
\]
where $p_\ell$ denotes the probability that the directed random walk on $G_1$
ends in~$\ell$.
\end{proposition}

Let us close this section by formally defining the class of galled trees
that we consider. Since in this paper we work in the context of DAGs, we use
the term \emph{connectivity} to refer to the notion of \emph{weak
connectivity} (that is, connectivity of the underlying undirected graph):
there is no risk of confusion with the notion of strong connectivity,
as this notion makes little sense for DAGs (whose strongly connected components
are reduced to single vertices).

Recall that $v$ is a \emph{cut-vertex} of $G$ if removing $v$
increases the number of connected components of $G$. A graph is said to be
\emph{biconnected} if it has no cut-vertices, and a \emph{biconnected
component} is a maximal biconnected subgraph. A biconnected component
consisting of a single vertex is said to be trivial.

Finally, a phylogenetic network $G$ is \emph{binary} if it is either reduced to a single
vertex, or if the root has out-degree~2; the leaves have in-degree~1; and every
other vertex has either in-degree~1 and out-degree~2 (\emph{tree-vertices}) or
in-degree 2 and out-degree~1 (\emph{reticulations}).

\begin{definition} \label{defGalledTrees}
A \emph{galled tree}, also known as a \emph{level-1 network}, is a
binary phylogenetic network with no multi-edge where each biconnected component
has at most one reticulation.
\end{definition}

The non-trivial biconnected components of a galled tree are known as its
\emph{galls}. They correspond to undirected cycles made of two directed paths
connecting a vertex, known as the \emph{root of the gall}, to a reticulation.
Note that the galls of a galled tree are, by definition, vertex-disjoint.

\begin{figure}[t]
  \centering
  \captionsetup{width=.8\linewidth}
  \includegraphics[width=0.4\linewidth]{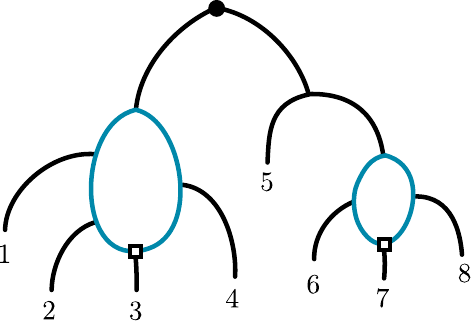}
  \caption{An example of a galled network with $n = 8$ labeled leaves and
  two galls.  The root is the black circle on top, and each point where three
  lines meet is an internal vertex. The direction of the edges is not
  indicated: they are always pointing downwards. The two galls are highlighted
  in blue, and the corresponding reticulations are denoted by white squares.
  The probability that the directed random walk ends in leaf~1 is $p_1 = 1/8$;
  the probability that it ends in leaf~$3$ is $p_3 = 1/16 + 1/8$.}
  \label{figIntroduction}
\end{figure}

Our definitions so far make no mention of labelings: they apply to unlabeled,
vertex-labeled and leaf-labeled networks alike. However, the labeling plays an
important role in the definition of our model of random phylogenetic network.

Let $\mathscr{G}_n$ denote the combinatorial class of galled trees with $n$
leaves labeled with the integers from $1$ to $n$. Note that
$\mathscr{G}_n$ is finite: the sequence $(\Card{\mathscr{G}_n})_{n\geq 1}$ is
registered as \href{https://oeis.org/A328122}{A328122} in
the On-Line Encyclopedia of Integer Sequences~\cite{OEIS},  and has
been studied extensively in \cite{bouvel2020counting},
where an explicit formula was obtained using analytic combinatorics.
Recently, Stufler~\cite{stufler2022branching} studied $\mathscr{G}_n$
(and, more generally, level-$k$ networks) using branching
processes; his method is the starting point of our probabilistic approach,
and we will come back to it later.

Since $\mathscr{G}_n$ is finite, we can endow it with the uniform distribution: 
we write $G_n \sim \mathrm{Unif}(\mathscr{G}_n)$ to
denote a phylogenetic network $G_n$ sampled uniformly at random
from~$\mathscr{G}_n$, and we refer to $G_n$ as a \emph{uniform galled tree with
$n$ labeled leaves}. 
Our main object of study in the rest of this document is the asymptotic
behavior of the random variable $B_2(G_n)$.

\section{Analytic combinatorics approach} \label{secAnalyticCombinatorics} 

In this section, we derive the asymptotics of moments of $B_2(G_n)$. To do so, we use a generating function approach that is based on the recursive decomposition of galled trees according to whether the root is in a gall or not; see Figure~\ref{gt-decomp}. This decomposition yields explicit expressions for the generating functions of moments, which can then be analysed with the method of singularity analysis from Analytic Combinatorics; see Chapter VI in~\cite{flajolet2009analytic}. The latter is based on so-called transfer theorems which say, in a nutshell, that if a generating function $f(z)$ is analytic in a suitable large domain $\Delta$ with a singularity $\xi\in{\mathbb R}^{+}$ at its boundary and $f(z)\sim c(1-z/\xi)^{-\alpha}$, as $z\rightarrow\xi$ in $\Delta$, then $[z^n]f(z)\sim cn^{\alpha-1}\xi^{-n}/\Gamma(\alpha)$, as $n\rightarrow\infty$, where $[z^n]f(z)$ denotes the $n$-th Taylor coefficient of $f(z)$ at $z=0$. More precisely, we use domains $\Delta$ of the form:
\[
\Delta \defas \Delta(r,\phi_0)=\{z\ :\ \vert z\vert\leq r,\ z\ne\xi,\ \vert\arg(z-\xi)\vert\geq \phi_0\}
\]
with $r>\vert\xi\vert$ and $0<\phi_0<\pi/2$. A function which is analytic in $\Delta$ (for some $r$ and $\phi_0$) is subsequently called {\it $\Delta$-analytic at $\xi$}.

\begin{figure}
  \centering
  \captionsetup{width=.9\linewidth}
  \includegraphics[width=0.8\linewidth]{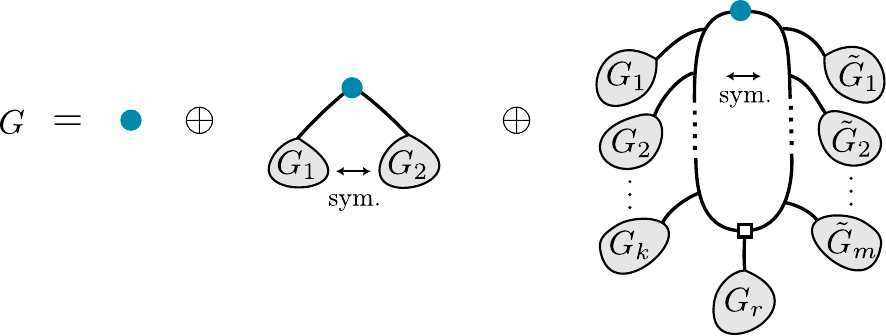}
  \caption{The recursive decomposition of a galled tree $G$. $G$ is either a
  leaf (first case after the equality sign) or a root which is not contained in
  a gall, with two galled trees $G_1$ and $G_2$ attached to it (second case);
  or the root itself is in a gall (third case; both sides of the gall are not
  allowed to be empty at the same time, i.e.\ $k+m\geq 1$).  Note that in the
  second and third cases, there is a symmetry: it does not
  matter which side of the gall is drawn to the left \emph{vs}  to the right.
  }\label{gt-decomp}
\end{figure}

We first need to recall some results from \cite{bouvel2020counting}.
Set $\mathscr{G}\defas\bigcup_{n\geq 0}\mathscr{G}_n$ which is the set of all galled trees. Moreover, set
\[
  {\rm GT}(z)\defas \sum_{G\in\mathscr{G}}\frac{z^{\vert G\vert}}{\vert G\vert!},
\]
which is the (exponential) generating function of the number of galled trees $G$ with $\vert G\vert \asdef n$ leaves. Then, from \cite{bouvel2020counting},
\begin{equation}\label{GTz}
{\rm GT}(z)=-\frac{\sqrt{2}\sqrt{-4\left(\sqrt{1-8z}-2\right)z+9\sqrt{1-8z}-1}}{4(1-8z)^{1/4}}-\frac{1}{4}\sqrt{1-8z}+\frac{5}{4}.
\end{equation}
Thus, for a $\Delta$-domain at $1/8$, ${\rm GT}(z)$ is $\Delta$-analytic with an expansion at $z=1/8$ of the form:
\begin{equation}\label{exp-GTz}
{\rm GT}(z)=\rho-\tau\sqrt{1-8z}+{\mathcal O}(1-8z),\qquad (z\rightarrow 1/8).
\end{equation}
The above limit is inside the $\Delta$-domain and
\[
\rho=\frac{5}{4}-\frac{\sqrt{17}}{4}\approx 0.219\qquad\text{and}\qquad\tau=\frac{1}{4}-\frac{\sqrt{17}}{68}\approx 0.189.
\]
From this, by the transfer theorems,
\[
\vert\mathscr{G}_n\vert=n![z^n]{\rm GT}(z)\sim -\tau n![z^n]\sqrt{1-8z}\sim\frac{\sqrt{2}\tau}{2}\left(\frac{8}{e}\right)^nn^{n-1};
\]
compare with Proposition 5.3 in \cite{bouvel2020counting}.

Formula (\ref{GTz}) was also derived with the decomposition mentioned above. (Actually, the authors in \cite{bouvel2020counting} divided the third case in Figure~\ref{gt-decomp} into two subcases, as this allowed an easier translation into generating function.) We now extend the analysis to the mean. 

Set:
\[
A(z)\defas \sum_{G\in\mathscr{G}}B_2(G)\frac{z^{\vert G\vert}}{\vert G\vert!},
\]
where $B_2(G)$ is the $B_2$ index of the galled tree $G$ whose definition we recall:
\[
B_2(G)=-\sum_{\ell\in \mathcal{L}(G)}p_{\ell}\log_2 p_{\ell},
\]
where $\mathcal{L}(G)$ denotes the leaf set of $G$ and $p_\ell$ is the
probability of reaching $\ell$.

First, consider the case where the root of $G$ is not in a gall.
Then, it has two galled subtrees $G_1,G_2$, and we have
\begin{equation}\label{case-1}
B_2(G)=\frac{1}{2}\left(B_2(G_1)+B_2(G_2)\right)+1
\end{equation}
as directly follows from the grafting property.

Next, consider the case where the root of $G$ is inside a gall. Assume that $G_1,\ldots,G_{k}$ and $\tilde{G}_1,\ldots,\tilde{G}_{m}$ are the galled subtrees attached to the left and right of the cycle (in any order) and $G_r$ the galled subtree below the reticulation vertex at the bottom of the gall. Then, for $k,m\geq 0$ with $k+m\geq 1$, we have
\begin{align}
B_2(G)&=-\sum_{s=1}^{k}\sum_{\ell\in\mathcal{L}(G_s)}\frac{p_{\ell}}{2^{s+1}}\log_2\left(\frac{p_{\ell}}{2^{s+1}}\right)-\sum_{t=1}^{m}\sum_{\ell\in \mathcal{L}(\tilde{G}_t)}\frac{p_{\ell}}{2^{t+1}}\log_2\left(\frac{p_{\ell}}{2^{t+1}}\right)\nonumber\\
&\qquad-\sum_{\ell\in\mathcal{L}(G_r)}\left(\frac{1}{2^{k+1}}+\frac{1}{2^{m+1}}\right)p_{\ell}\log_2\left(\left(\frac{1}{2^{k+1}}+\frac{1}{2^{m+1}}\right)p_{\ell}\right)\nonumber\\
&=\sum_{s=1}^{k}\frac{B_2(G_s)}{2^{s+1}}+\sum_{t=1}^{m}\frac{B_2(\tilde{G}_t)}{2^{t+1}}+\left(\frac{1}{2^{k+1}}+\frac{1}{2^{m+1}}\right)B_2(G_r)\nonumber\\
&\qquad+\sum_{s=1}^{k}\frac{s+1}{2^{s+1}}+\sum_{t=1}^{m}\frac{t+1}{2^{t+1}}-\left(\frac{1}{2^{k+1}}+\frac{1}{2^{m+1}}\right)\log_2\left(\frac{1}{2^{k+1}}+\frac{1}{2^{m+1}}\right)\nonumber\\
&=\sum_{s=1}^{k}\frac{B_2(G_s)}{2^{s+1}}+\sum_{t=1}^{m}\frac{B_2(\tilde{G}_t)}{2^{t+1}}+\left(\frac{1}{2^{k+1}}+\frac{1}{2^{m+1}}\right)B_2(G_r)\nonumber\\
&\qquad+3-\frac{k+3}{2^{k+1}}-\frac{m+3}{2^{m+1}}-\left(\frac{1}{2^{k+1}}+\frac{1}{2^{m+1}}\right)\log_2\left(\frac{1}{2^{k+1}}+\frac{1}{2^{m+1}}\right).\label{case-3}
\end{align}

We now use the recursive decomposition, which gives:
\begin{align}
A(z)&=\frac{1}{2}\sum_{G_1,G_2}B_2(G)\frac{z^{\vert G_1\vert+\vert G_2\vert}}{\vert G_1\vert!\vert G_2\vert!}\nonumber\\
&\qquad+\frac{1}{2}\sum_{\substack{k,m\geq 0\\ k+m\geq 1}}\sum_{\substack{G_1,\ldots,G_k \\ \tilde{G}_1,\ldots,\tilde{G}_m \\ G_r}} B_2(G)\frac{z^{\vert G_1\vert+\cdots+\vert G_k\vert+\vert\tilde{G}_1\vert+\cdots+\vert\tilde{G}_m\vert+\vert G_r\vert}}{\vert G_1\vert!\cdots\vert G_k\vert!\vert\tilde{G}_1\vert!\cdots\vert\tilde{G}_m\vert!\vert G_r\vert!},\label{equ-Az}
\end{align}
where inside the first and second sum, we have to replace $B_2(G)$ by (\ref{case-1}) and (\ref{case-3}), respectively. This gives, for the first sum,
\[
\sum_{G_1,G_2}B_2(G)\frac{z^{\vert G_1\vert+\vert G_2\vert}}{\vert G_1\vert!\vert G_2\vert!}=A(z){\rm GT}(z)+{\rm GT}(z)^2.
\]
For the second sum, we have
\begin{align*}
&\sum_{\substack{k,m\geq 0\\ k+m\geq 1}}\sum_{\substack{G_1,\ldots,G_k \\ \tilde{G}_1,\ldots,\tilde{G}_m \\ G_r}} B_2(G)\frac{z^{\vert G_1\vert+\cdots+\vert G_k\vert+\vert\tilde{G}_1\vert+\cdots+\vert\tilde{G}_m\vert+\vert G_r\vert}}{\vert G_1\vert!\cdots\vert G_k\vert!\vert\tilde{G}_1\vert!\cdots\vert\tilde{G}_m\vert!\vert G_r\vert!}\\
&\qquad=A(z)\sum_{\substack{k,m\geq 0\\ k+m\geq 1}}\left(\sum_{s=1}^{k}\frac{1}{2^{s+1}}+\sum_{t=1}^{m}\frac{1}{2^{t+1}}+\left(\frac{1}{2^{k+1}}+\frac{1}{2^{m+1}}\right)\right){\rm GT}(z)^{k+m}\\
&\qquad\quad+\sum_{\substack{k,m\geq 0\\ k+m\geq 1}}\left(3-\frac{k+3}{2^{k+1}}-\frac{m+3}{2^{m+1}}\right){\rm GT}(z)^{k+m+1}\\
&\qquad\quad-\sum_{\substack{k,m\geq 0\\ k+m\geq 1}}\left(\frac{1}{2^{k+1}}+\frac{1}{2^{m+1}}\right)\log_2\left(\frac{1}{2^{k+1}}+\frac{1}{2^{m+1}}\right){\rm GT}(z)^{k+m+1}\\
&\qquad=A(z)\sum_{\substack{k,m\geq 0\\ k+m\geq 1}}{\rm GT}(z)^{k+m}+\sum_{k,m\geq 0}\left(3-\frac{k+3}{2^{k+1}}-\frac{m+3}{2^{m+1}}\right){\rm GT}(z)^{k+m+1}\\
&\qquad\quad-\sum_{k,m\geq 0}\left(\frac{1}{2^{k+1}}+\frac{1}{2^{m+1}}\right)\log_2\left(\frac{1}{2^{k+1}}+\frac{1}{2^{m+1}}\right){\rm GT}(z)^{k+m+1}\\
&\qquad=A(z)\left(\frac{1}{({\rm GT}(z)-1)^2}-1\right)-\frac{{\rm GT}(z)^2({\rm GT}(z)-4)}{({\rm GT}(z)-1)^2({\rm GT}(z)-2)^2}-h({\rm GT}(z)),
\end{align*}
where 
\[
h(\omega)=\sum_{k,m\geq 0}\left(\frac{1}{2^{k+1}}+\frac{1}{2^{m+1}}\right)\log_2\left(\frac{1}{2^{k+1}}+\frac{1}{2^{m+1}}\right)\omega^{k+m+1}.
\]

Plugging everything into (\ref{equ-Az}) and solving for $A(z)$ gives
\[
A(z)=\frac{f({\rm GT}(z))}{g({\rm GT}(z))},
\]
where
\[
g(\omega)=1-\frac{\omega}{2}+\frac{\omega(\omega-2)}{2(\omega-1)^2}
\]
and
\[
f(\omega)=\frac{\omega^2}{2}-\frac{\omega^2(\omega-4)}{2(\omega-1)^2(\omega-2)^2}-\frac{h(\omega)}{2}.
\]

We next show that $A(z)$ is $\Delta$-analytic (with a $\Delta$-domain at $1/8$) and derive its expansion as $z\rightarrow 1/8$. First note that $f(\omega)/g(\omega)$ is analytic in $\vert\omega\vert<1$. Expanding at $\rho\approx 0.219$ gives
\[
\frac{f(\omega)}{g(\omega)}=\frac{f(\rho)}{g(\rho)}+\frac{f'(\rho)g(\rho)-f(\rho)g'(\rho)}{g(\rho)^2}(\omega-\rho)+{\mathcal O}((\omega-\rho)^2).
\]
By plugging into this (\ref{exp-GTz}),
\[
A(z)=\frac{f(\rho)}{g(\rho)}+\frac{f(\rho)g'(\rho)-f'(\rho)g(\rho)}{g(\rho)^2}\tau\sqrt{1-8z}+{\mathcal O}((1-8z)),\qquad (z\rightarrow 1/8).
\]
Thus, by the transfer theorems,
\[
[z^n]A(z)\sim\frac{f(\rho)g'(\rho)-f'(\rho)g(\rho)}{g(\rho)^2}\tau[z^n]\sqrt{1-8z}
\]
which implies for the mean of the $B_2$ index of a random galled tree
$G_n \sim \mathrm{Unif}(\mathscr{G}_n)$:
\[
\ExpecBrackets{B_2(G_n)}=\frac{[z^n]A(z)}{[z^n]{\rm GT}(z)}\sim\frac{f'(\rho)g(\rho)-f(\rho)g'(\rho)}{g(\rho)^2}=2.707911858984\ldots,
\]
where the numerical evaluation was done with Maple. We summarize this in our first result of this section.
\begin{theorem} \label{thmLimExpecB2AnalyticCombinatorics}
The mean value of the $B_2$ index of a uniform galled tree with $n$ labeled
leaves converges to $2.707911858984\ldots$.
\end{theorem}

The above computations can be generalized to higher moments. More precisely, define
\[
  A^{[\ell]}(z) \defas \sum_{G\in\mathscr{G}}B_2(G)^{\ell}\frac{z^{\vert G\vert}}{\vert G\vert!}
\]
so that $A^{[0]}(z)={\rm GT}(z)$ and $A^{[1]}(z)=A(z)$. Then, by raising (\ref{case-1}) and (\ref{case-3}) to the $\ell$-th power, plugging them into (\ref{equ-Az}), and expanding, we obtain that
\[
A^{[\ell]}(z)=\frac{f_{\ell}(z)}{g_{\ell}({\rm GT}(z))},
\]
where
\begin{align*}
  g_{\ell}(\omega) \defas 1-\frac{\omega}{2^{\ell}}&-\sum_{k\geq 1}\left(\sum_{s=1}^{k}\frac{1}{2^{(s+1)\ell}}+\left(\frac{1}{2}+\frac{1}{2^{k+1}}\right)^{\ell}\right)\omega^k\\
&-\frac{1}{2}\sum_{k,m\geq 1}\left(\sum_{s=1}^{k}\frac{1}{2^{(s+1)\ell}}+\sum_{t=1}^{m}\frac{1}{2^{(t+1)\ell}}+\left(\frac{1}{2^{k+1}}+\frac{1}{2^{m+1}}\right)^{\ell}\right)\omega^{k+m}
\end{align*}
and

\break

\begin{align}
  f_{\ell}(z) \defas\; &\frac{1}{2}\sum_{\substack{\ell_1+\ell_2+\ell_3=\ell\\\ell_1,\ell_2<\ell}}\binom{\ell}{\ell_1,\ell_2,\ell_3}\frac{1}{2^{\ell_1+\ell_2}}A^{[\ell_1]}(z)A^{[\ell_2]}(z)\nonumber\\
&+\frac{1}{2}\sum_{\substack{k,m\geq 0\\k+m\geq 1}}\sum_{\substack{\ell_1+\cdots+\ell_{k+m+2}=\ell\\\ell_j<\ell,1\leq j\leq k+m+1}}\binom{\ell}{\ell_1,\cdots,\ell_{k+m+2}}\prod_{s=1}^{k}\frac{1}{2^{(s+1)\ell_s}}\prod_{t=1}^{m}\frac{1}{2^{(t+1)\ell_{k+t}}}\nonumber\\
&\qquad\qquad\times\left(\frac{1}{2^{k+1}}+\frac{1}{2^{m+1}}\right)^{\ell_{k+m+1}}\mu_{k,m}^{\ell_{k+m+2}}\prod_{s=1}^{k}A^{[\ell_s]}(z)\prod_{t=1}^{m+1}A^{[\ell_{k+t}]}(z),\label{flz}
\end{align}
where
\begin{equation}\label{mukm}
  \mu_{k,m} \defas 3-\frac{k+3}{2^{k+1}}-\frac{m+3}{2^{m+1}}-\left(\frac{1}{2^{k+1}}+\frac{1}{2^{m+1}}\right)\log_2\left(\frac{1}{2^{k+1}}+\frac{1}{2^{m+1}}\right).
\end{equation}

Now, define two sequences $c_{\ell}$ and $d_{\ell}$ by
\[
A^{[\ell]}(z)=c_{\ell}-d_{\ell}\tau\sqrt{1-8z}+\cdots
\]
Then, we have $c_{0}=\rho, d_{0}=1$, and for $\ell\geq 1$,
\[
c_{\ell}=\frac{f_{\ell}(A^{[j]}(z)\leftrightarrow c_j)}{g_{\ell}(\rho)}
\]
and
\begin{equation}\label{rec-d-ell}
d_{\ell}=\frac{f'_{\ell}((A^{[j]})'(z)\leftrightarrow d_j,A^{[j]}(z)\leftrightarrow c_j)}{g_{\ell}(\rho)}-c_{\ell}\frac{g_{\ell}'(\rho)}{g_{\ell}(\rho)},
\end{equation}
where $f_{\ell}(A^{[j]}(z)\leftrightarrow c_j)$ means that in $f_{\ell}(z)$ we replace $A^{[j]}(z)$ with $j<\ell$ by $c_j$ and similar for the numerator of the first fraction on the right-hand side of (\ref{rec-d-ell}).

Finally, by applying singularity analysis as above, we obtain that
\begin{equation}\label{asym-lth-moment}
\ExpecBrackets{B_2(G_n)^{\ell}}\rightarrow d_{\ell}.
\end{equation}
The above recurrence for $d_{\ell}$ can be, e.g., used to compute $d_2=7.965561677362\ldots$ (again, e.g., with Maple). Thus, we have the following result for the variance.
\begin{theorem}
The variance of the $B_2$ index of a uniform galled tree with $n$ labeled
leaves converges to $0.632774946963\ldots$.
\end{theorem}

Also, from (\ref{asym-lth-moment}) and the method of moments, we can identify the limit law.
\begin{theorem}\label{ll-B2}
There exists a random variable $B$ with $\ExpecBrackets*{\normalsize}{B^{\ell}}=d_{\ell}$
whose distribution is the limit distribution of the $B_2$ index of a
uniform galled tree with $n$ labeled leaves, i.e.,
\[
B_2(G_n)\stackrel{d}{\longrightarrow} B,
\]
where the convergence also holds for all moments.
\end{theorem}

\begin{remark}
An explicit construction of the random variable $B$ from Theorem~\ref{ll-B2}
is given in Section~\ref{secLocalLimit}, where it is obtained 
as the $B_2$ index of the infinite phylogenetic network $G_*$ corresponding
to the local limit of uniform leaf-labeled galled trees.
\end{remark}

\begin{proof}
According to classical results from probability theory (see \cite[Section~30]{billingsley1995probability}), we only have to show
that the exponential generating function of $\sum_{\ell\geq 0}d_{\ell}z^{\ell}/\ell!$
has a non-zero radius of convergence. This follows from an estimate of the form
\begin{equation}\label{upper-bound-dl}
d_{\ell}\leq K^{\ell}\ell!,
\end{equation}
where $K>0$ is a suitable constant. Such an estimate is derived from the recurrence for $d_{\ell}$ and induction; see Appendix~\ref{ll-mom}. 
\end{proof}

\section{Local limit approach} \label{secLocalLimit}

\subsection{General principle}

Local limits make it possible to capture certain aspects of the structure of
large graphs. Typically, one considers a sequence $(G_n)$ of finite graphs of
increasing size.  The idea is then to define a (potentially infinite) rooted
graph $G_\infty = \lim_n G_n$ whose structure around the root tends to match
the local structure of $G_n$ ``as seen from a focal vertex'' when $n$ goes to
infinity. This focal vertex can be fixed or random, depending on the
application.

Let us be more specific and briefly recall the definition of local convergence.
First, we need to introduce a notion of \emph{restriction} of a phylogenetic
network.

\begin{definition} \label{defRestriction}
Let $G$ be a phylogenetic network and let $v$ be a vertex of $G$.
The \emph{height} of $v$, which we denote by $h(v)$, is the number of edges
of a shortest path going from the root of $G$ to $v$.
The \emph{restriction of height $k$} of $G$, which we denote by
$[G]_k$, is the subgraph of $G$ induced by its vertices of height at most $k$.
\end{definition}

Note that, because phylogenetic networks are locally finite, $[G]_k$ is
finite for every phylogenetic network $G$ and every $k\geq 0$. Therefore,
it is natural to endow the set of finite phylogenetic networks with the
discrete topology and thus say that $[G_n]_k \to [G]_k$ when there exists $N$ such
that for all $n \geq N$, $[G_n]_k = [G]_k$.

\begin{definition} \label{defLocalCV}
A sequence $(G_n)$ of phylogenetic networks \emph{converges locally} to the
phylogenetic network $G$ if, for all fixed $k \geq 0$, as $n$ goes to infinity,
$[G_n]_k$ converges to $[G]_k$ in the discrete topology.
\end{definition}

\begin{remark}
Note that local convergence is metrizable: the metric $d$
defined by $d(G,G') = \exp(-\min\{k \geq 0 : [G]_k \neq [G']_k\})$, with
the convention $\min\emptyset = +\infty$, generates the
so-called \emph{local topology}, and $G_n$ converges locally to $G$ if
and only if $d(G_n, G)\to 0$.
Furthermore, the space $\phylnet$ of phylogenetic network  equipped with the
metric $d$ is separable and complete.
\end{remark}

Local convergence has now become a standard tool in probability
theory, and we refer the reader to~\cite[Chapter~2]{vanderhofstad2024random}
for a detailed introduction. Its usefulness comes in great part from the fact
that the local limit $G_\infty$ of a sequence $(G_n)$ of graphs is often more
tractable than the finite graphs $G_n$.  In particular, regions of $G_n$ that
are only ``asymptotically independent'' can become truly independent
in~$G_\infty$. As a result, given a functional $f$ of interest, it is sometimes
much easier to compute $f(G_\infty)$ than $f(G_n)$ and this can yield a simple
way to compute $\lim_n f(G_n)$. However, for this one must:
\begin{itemize}
  \item Identify the limit $G_\infty$ as a tractable graph.
  \item Ensure that $f$ is well-defined for infinite graphs, and that it
    is continuous for the local topology (or at least along the sequence of
    interest).
\end{itemize}

In the case of galled trees, the local limit has been identified
in~\cite{stufler2022branching} and is indeed highly tractable -- see
Section~\ref{secBlowUps}. However, the extension of the $B_2$ index
to infinite phylogenetic networks and its continuity are challenging.

\subsection{The \texorpdfstring{$B_2$}{B2} index of infinite phylogenetic networks} \label{secB2InfiniteNetwork}

The Definition~\ref{defB2Finite} of the $B_2$ index of a finite phylogenetic
network as the Shannon entropy of the probability distribution induced on its
leaves by the directed random walk $X$ does not immediately
extend to infinite networks. This is because the random walk can escape to
infinity, without ever ending in a leaf (in fact, note that an infinite
phylogenetic network may not even have leaves).

In \cite[Definition~1.4]{bienvenu2021revisiting}, it was claimed that a simple
way to circumvent this problem is to define the $B_2$ index of an infinite
phylogenetic network $G$ as $\lim_k B_2([G]_k)$. However, while this works for
trees (Proposition~\ref{propStubSeq} in Appendix~\ref{appPropB2}), 
this is not the case for general phylogenetic networks, because the sequence
$B_2([G]_k)$ may not converge, as the simple example given in
Figure~\ref{figCounterExDefB2} shows.

\begin{figure}[h]
  \centering
  \captionsetup{width=.75\linewidth}
  \includegraphics[width=0.4\linewidth]{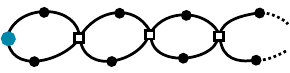}
  \caption{A phylogenetic network $G$ illustrating why the
   Definition~1.4 from~\cite{bienvenu2021revisiting} may fail. The root is
   the vertex highlighted in blue on the right, and the white squares denote
   reticulations. Here, $B_2([G]_k) = \Indic{k \text{ is odd}}$ does not
   converge.}
  \label{figCounterExDefB2}
\end{figure}

To get a general definition, we define a suitable notion of ``boundary'' of
DAGs on which the random walk $X$ induces a probability
distribution. In the finite case, this boundary is simply the leaf set; but in
the infinite case it consists of the leaf set and of all ``directions'' in
which $X$ can escape.  While there already exist several notions
of boundaries to quantify the asymptotic behavior of random walks, such as the
Martin boundary \cite[Chapter~10]{kemeny1976denumerable}, we are unaware of a
notion of boundary that matches ours.  Moreover, here we give an
explicit construction that is useful to understand the properties of the $B_2$
index.

We now give the formal definition of the $B_2$ index of an arbitrary
phylogenetic network, and list some of its most important properties.
First, we recall the notion of \emph{ends} of an infinite graph.
Note that here we are working with directed graphs, and therefore that the
notion of end that we use differs from the more common notion of end of an
undirected graph.

A \emph{ray} is a one-way infinite directed path $v_0 \to v_1  \to \cdots$.
Two rays $r$ and $r'$ are said to be \emph{co-directional}, which we denote by
$r \codir r'$, if there exists a ray $r''$, which could be one of $r$ or $r'$,
that intersects $r$ and $r'$ an infinite number of times.
It is readily checked that $\codir$ is an equivalence relation.

\begin{definition} \label{defEnds}
Let $G$ be a DAG with at least one ray.  The equivalence classes of the
co-directional relation $\codir$ are called the \emph{ends} of $G$.
\end{definition}

An infinite DAG that is not locally finite may or may not have ends.
However, an infinite phylogenetic network, being locally finite, must have at
least one end (since it has at least one ray).
Moreover, for any $\kappa \in \N \cup \{\aleph_0, 2^{\aleph_0}\}$, there
exist phylogenetic networks with $\kappa$ ends,
illustrated in Figure~\ref{figExamplesEnds}.

\begin{figure}[h]
  \centering
  \captionsetup{width=.95\linewidth}
  \includegraphics[width=0.95\linewidth]{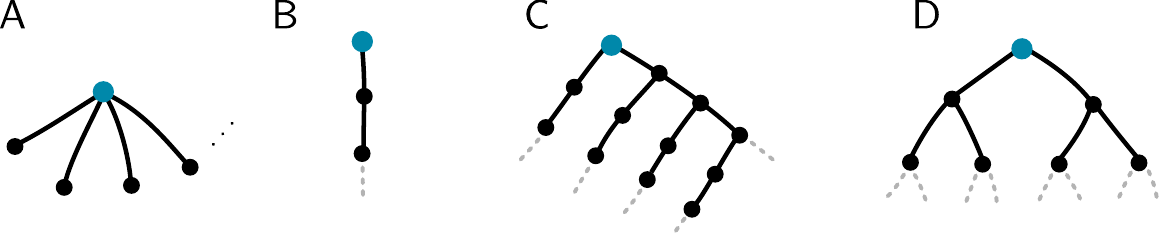}
  \caption{Examples of infinite DAGs with various number of ends. The root is
  highlighted in blue, and the direction of the edges is not indicated, as they
  are always oriented away from the root.  \textsf{A}: the infinite star graph,
  an infinite DAG with no ends.  This DAG is not locally finite, and therefore
  is not a phylogenetic network. \textsf{B}: the semi-infinite path, a DAG with
  one end. \textsf{C}: the DAG obtained by grafting infinite paths to the
  leaves of the infinite caterpillar. This DAG has $\aleph_0 = \Card(\N)$ ends.
  \textsf{D}: the infinite complete binary tree, a DAG with $2^{\aleph_0} =
  \Card(\R)$ ends.}
  \label{figExamplesEnds}
\end{figure}

\begin{definition} \label{defBoundary}
Let $G$ be phylogenetic network with leaf set $\mathcal{L}$ and end
set $\mathcal{E}$. The \emph{boundary} of $G$ is the set
$\partial G = \mathcal{L} \cup \mathcal{E}$.
\end{definition}

In Appendix~\ref{appBoundary}, we prove that the limit of the random walk
$(X_t)_{t\geq 0}$ is a well-defined random variable $X_\infty$ taking values in
$\partial G$.  We denote its distribution by $\mu$. Using the usual definition
of the Shannon entropy $H$ of an arbitrary probability distribution (see
Appendix~\ref{appPropEntropy}), this makes it possible to extend
Definition~\ref{defB2Finite} to infinite phylogenetic networks.

\enlargethispage{1ex}
\begin{definition} \label{defB2Infinite}
Let $G$ be a phylogenetic network. The \emph{$B_2$ index of $G$} 
is the Shannon entropy of the probability distribution $\mu$ induced on
the boundary of $G$ by the directed random walk: $B_2(G) = H(\mu)$.
\end{definition}

The properties of the $B_2$ index of infinite phylogenetic networks are
essentially the same as in the finite case, except that $B_2$ can now equal
$+\infty$. In particular, the grafting property still holds for infinite
phylogenetic networks.  This is proved in Appendix~\ref{appPropB2}, where other
basic properties of the $B_2$ index of general phylogenetic networks are
provided.

There is, however, one major difficulty with our approach: the $B_2$ index is
not continuous for the local topology. In fact, $B_2$ is nowhere continuous on
the set of infinite phylogenetic networks -- in the sense that for every
infinite phylogenetic network $G$ there exists a sequence $(G_n)$ such that
$G_n \to G$ for the local topology but $B_2(G_n) \not\to B_2(G)$, see
Figure~\ref{figNotContinuous}.  This is because the $B_2$ index reflects the
structure of the boundary of a network, and parts of the boundary can be
located infinitely far away from the root; whereas functions that are
continuous for the local topology must, by definition, depend on the local
structure of the network around the root.

\begin{figure}[ht]
  \centering
  \captionsetup{width=.8\linewidth}
  \includegraphics[width=0.25\linewidth]{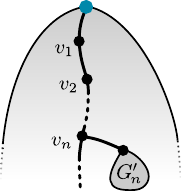}
  \caption{Let $G$ be an infinite phylogenetic network, and let $v_1 \to v_2
  \to \cdots$ be a ray of $G$ (being infinite and locally finite,
  $G$ always has at least one ray). Pick a sequence $(G'_n)$ of
  phylogenetic networks, and let $G_n$ be the phylogenetic network obtained by
  grafting $G'_n$ on $v_n$. Thus, $G_n \to G$. However, by the grafting property,
  $B_2(G_n) = B_2(G) + p_{v_n} B_2(G'_n)$ can be arbitrarily large.}
  \label{figNotContinuous}
\end{figure}

Nevertheless, it is still possible to use local limits to study the $B_2$ index:
indeed, $B_2$ can be continuous along sequences $(G_n)$ of interest -- in fact,
we expect that this should be the case for many biologically relevant models of
phylogenetic networks.

In Appendix~\ref{appPropB2}, we provide general
results to study the continuity of $B_2$. For instance,
Proposition~\ref{propInequalityLimInfB2} states that if the local limit $G$ of
a sequence $(G_n)$ is such that the directed random walk $X$ gets trapped in a
leaf with probability~1, then
\[
   \liminf B_2(G_n) \;\geq\; B_2(G)  \,.
\]
Similarly, in Proposition~\ref{propInequalityLimSupB2} we provide an
easy-to-check sufficient condition ensuring that $\lim_n B_2(G_n) = B_2(G)$.

In the next section, we state our main continuity result (on which our study
of the $B_2$ index of uniform leaf-labeled galled trees relies): that
the $B_2$ index is essentially continuous for a generic class of models of
random phylogenetic networks known as blowups of Galton--Watson trees.

\subsection{Blowups of Galton--Watson trees} \label{secBlowUps}

Informally, a blowup of a random rooted tree $T$ is a random phylogenetic
network that is obtained by: first, sampling $T$; then replacing each
internal vertex~$v$ by an independent realization $\Head{v}$ of a random
phylogenetic network (identifying~$v$ with the root of $\Head{v}$, and each
of the children of $v$ with a leaf of $\Head{v}$ -- see
Figure~\ref{figExampleBlowup}).  In this construction, we require that
$\Head{v}$ depend on $d^+(v)$ only, where $d^+(v)$ denotes the number of
children of $v$ in $T$, and that the matching between the leaves of $\Head{v}$
and the children of $v$ be chosen uniformly at random and independently of
everything else.

Formally, let $T$ be a random ordered tree -- that is, $T$ is rooted and the
children of each vertex $v \in T$ are ordered from $1$ to $d^+(v)$ -- which we
view as a subset of the Ulam--Harris tree $\mathcal{U}$ (see e.g.\
\cite[Section~6]{janson2012simply}).  Let $\nu = (\nu_{k})_{k \geq 1}$, where
each $\nu_k$ is a leaf-exchangeable probability distribution on the set of
finite phylogenetic networks with $k$ leaves labeled $1, \ldots, k$ (by
\emph{leaf-exchangeable}, we mean invariant under permutation of the labels of
the leaves). Finally, let $\Gamma = (\Head{v}^{k} : v \in \mathcal{U}, k \geq 1)$
be a family of independent phylogenetic networks such that
$\Head{v}^{k} \sim \nu_{k}$ and $\Gamma$ is independent of $T$.

The \emph{blowup of $\,T\!$ with respect to $\nu$} is the random phylogenetic
network obtained from $(T, \Gamma)$ by gluing the networks
$\Head{v} \defas \Head{v}^{d^+\!(v)}$, where $v$ ranges over the internal
vertices of $T$, as follows: for each internal vertex $v \in T$, let $u_k$
denote the \mbox{$k$-th} child of $v$ and, for each non-leaf $u_k$, identify the
root of $\Head{u_k}$ with the leaf of $\Head{v}$ labeled~$k$.
Note that the leaves of the resulting network $G$ are not properly labeled, but
that they are in bijection with the leaves of $T$ (which are canonically
labeled by the Ulam--Harris labeling). To get a leaf-labeled network, one
can label the leaves of~$T$ (e.g, uniformly at random) and carry over the
labels to $G$; this is irrelevant when considering label-invariant functions
of~$G$, such as the $B_2$ index.

\begin{figure}[h!]
  \centering
  \captionsetup{width=.95\linewidth}
  \includegraphics[width=0.95\linewidth]{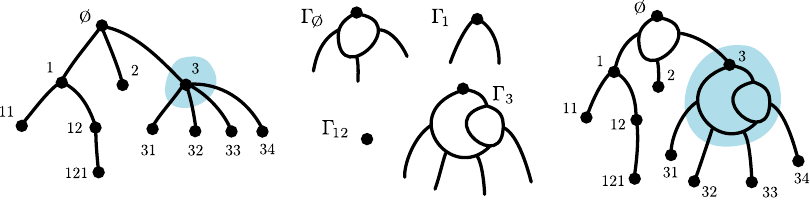}
  \caption{Illustration of the blowup: left, 
   an ordered tree $T$ with its Ulam--Harris labeling; center, 
   the random networks associated to the internal vertices of $T$;
   right, the network $G$ resulting from the blowup (with the Ulam--Harris
   labeling of the vertices of $T$ indicated on the corresponding vertices of
   $G$, to make the link between $T$ and $G$ more apparent). Each vertex
   $v \in T$ has been ``replaced'' by the network~$\Head{v}$; for instance,
   vertex~$3$ corresponds to the subgraph of $G$ highlighted in blue,
   which is isomorphic to $\Head{3}$.}
  \label{figExampleBlowup}
\end{figure}

Blowups of Galton--Watson trees are a noteworthy class of random phylogenetic
networks. Indeed,
\begin{itemize}
  \item Prominent ``combinatorial'' models of random phylogenetic networks --
    i.e.\ models that correspond to the uniform distribution on some relevant
    class of networks -- can be obtained as blowups of (size-conditioned)
    Galton--Watson trees.  For instance, in \cite{stufler2022branching} Stufler
    proved that this is the case for uniform leaf-labeled level-$k$ networks
    (and, therefore, for the galled trees considered here, as detailed in
    Section~\ref{secAppLocLimGalledTrees} below).  Blowups of Galton--Watson
    trees can also result from biologically relevant evolutionary processes,
    see~\cite{bienvenu2024branching}.
  \item The constrained structure of blowups of Galton--Watson trees makes them
    highly tractable. In \cite{stufler2022branching}, this was used
    to obtain asymptotic counting results for level-$k$ networks,
    as well as an explicit description of their large-scale geometry.
    Another example illustrating these constraints is given
    in~\cite{bienvenu2024zero}, where it is proved that if a combinatorial
    model of phylogenetic networks can be obtained as a blowup of a
    Galton--Watson tree, then for any fixed subgraph $S$ the fraction of the
    networks that contain $S$ in the corresponding combinatorial class is
    either~0 or~1.
\end{itemize}

The tractability of blowups of Galton--Watson trees comes, with varying degrees
of directness, from that of the underlying trees. In particular, the remarkable
asymptotic behavior of large size-conditioned critical Galton--Watson trees
often plays a crucial role. Here, we only focus on their local limit -- which,
as we will see, has a universal structure known as \emph{Kesten's tree}.

\enlargethispage{1ex}

\begin{definition} \label{defKesten}
Let $\eta = (\eta_i)_{i \geq 0}$ be a probability distribution on the
integers such that $\eta_0 > 0$ and $\sum_{i \geq 0} i \eta_i = 1$.
The \emph{Kesten tree associated to $\eta$} is the
two-type (\emph{spine}/\emph{regular}) Galton--Watson tree $T_*$ such that:
  \begin{itemize}
    \item \emph{Regular} vertices have offspring distribution $\eta$, and all
      of their children are regular vertices.
    \item \emph{Spine} vertices have offspring distribution $\hat{\eta}$ given
      by $\hat{\eta}_i = i \eta_i$, and exactly one of their children (whose
      order is chosen uniformly) is a spine vertex.
    \item The root of $T_*$ is a spine vertex.
  \end{itemize}
A Kesten tree $T_*$ is always infinite, and has exactly one ray 
$v_0\to v_1\to \dots$ starting from the root. This ray is known as the \emph{spine}.
\end{definition}

Let $T$ be a Galton--Watson tree with critical offspring distribution $\eta$
such that $\eta_0 > 0$.  For all $n$ such that $\Prob{\Abs{T}=n}>0$, where
$\Abs{T}$ denotes the number of leaves of~$T$, let $T_n$ denote a random tree
distributed as $T$ conditioned to have $n$ leaves. It is well-known that $T_n
\to T_*$ in distribution for the local topology, where $T_*$ is the Kesten
tree associated to $\eta$; see e.g.\ \cite[Theorem~7.1]{janson2012simply}.

Let now $G$ be a blowup of $T$ with respect to some family~$\nu$ of
leaf-exchangeable finite random graphs, and let $G_n$ be
distributed as $G$ conditioned to have $n$ leaves, noting that
$G_n$ can equivalently be described as a blowup of $T_n$ with respect to~$\nu$.
Because the blowup procedure is a local operation, the convergence of
$T_n$ to $T_*$ implies that $G_n \to G_*$ in distribution for the local
topology, where $G_*$ is a blowup of $T_*$ with respect to~$\nu$.

The following theorem, which we prove in Appendix~\ref{appProofsBlowUps},
is a continuity criterion for the $B_2$ index of blowups of Galton--Watson
trees.

\begin{theorem} \label{thmContinuityB2BlowUps}
With the notation above, assuming that the offspring distribution~$\eta$ is
critical, satisfies $\eta_0>0$ and has a finite third moment, we have:
\begin{mathlist}
  \item $B_2(G_n) \to B_2(G_*)$ in distribution.
  \item For all $p\geq 1$,
    $\ExpecBrackets{B_2(G_n)^p} \to \ExpecBrackets{B_2(G_*)^p}$,
    and all these moments are finite.
\end{mathlist}
\end{theorem}

Note the assumption in Theorem~\ref{thmContinuityB2BlowUps} that the offspring
distribution~$\eta$ has a finite third moment. This is because,
as already pointed out in Section~\ref{secB2InfiniteNetwork},
the mere convergence of $G_n$ to $G_*$ is not sufficient to get the convergence
of~$B_2(G_n)$: we also need to control the speed of convergence. Loosely speaking,
we need an upper bound on the total variation distance between
$[G_n]_{k_n}$ and $[G_*]_{k_n}$, for suitable sequences $(k_n)$.
This upper bound comes from a total variation bound for the convergence
of Galton--Watson trees: see Proposition~\ref{propDTV} in
Appendix~\ref{appTotalVariationBound}.
This total variation bound relies on a Berry--Esseen type local central limit
theorem, hence the third moment condition.

Combined with the independence between the base tree and the random networks
used in the blowup procedure, the recursive structure of Galton--Watson
trees makes it possible to get a simple expression for the expected value
of the $B_2$ index.

\begin{theorem} \label{thmExpecB2BlowUpGW}
Let $T$ be a Galton--Watson tree whose offspring distribution $\eta$ is such that
$\eta_0 > 0$, and let $G$ be a blowup of $T$ with respect to $\nu$. Then,
\[
  \ExpecBrackets{B_2(G)}  \;=\; \frac{1}{\eta_0} \, \ExpecBrackets{f(\xi)}\,, 
\]
where $\xi \sim \eta$, $f(0) = 0$ and, for $k \geq 1$,
$f(k) = \ExpecBrackets{B_2(\Head{k})}$, where 
$\Head{k} \sim \nu_k$.
If in addition $\eta$ has mean~1, then denoting by $T_*$ the Kesten tree
associated to $T$ and by $G_*$ the blowup of $T_*$ with respect to~$\nu$,
\[
  \ExpecBrackets{B_2(G_*)} \;=\; \frac{1}{\eta_0}
  \big(\ExpecBrackets{f(\xi)} \;+\; \ExpecBrackets*{\normalsize}{f(\hat{\xi}\,)}\big)\,,
\]
where $\smash{\hat{\xi}}\sim \hat{\eta}$, the size-biased distribution
associated to $\eta$ -- that is, $\hat{\eta}_k = k\, \eta_k$.
\end{theorem}

\begin{proof}
First, we use the recursive structure of~$G$ (see Figure~\ref{figLocalLimitGalled})
to characterize the distribution of~$B_2(G)$. For this, note that:
\begin{itemize}
\item with probability $\eta_0$, $G$ is reduced to its root, and $B_2(G) = 0$;
\item with probability $\eta_k$, the root of $T$ has $k$ children and
  thus the network $\Gamma$ associated with it is distributed as $\nu_k$.
  In that case, by the grafting property,
  \begin{equation} \label{eqRecursionDistriB2}
    B_2(G) \;\mathrel{\overset{d}{=}}\;
    B_2(\Gamma) \;+\; \sum_{i = 1}^k q_i\,B_2(G^{(i)}) \,,
  \end{equation}
  where $G^{(1)}, \ldots, G^{(k)}$ are independent copies of $G$, and
  $q_i$ is the probability that the directed random walk started from the
  root of $\Gamma$ reaches its leaf labeled~$i$.
\end{itemize}

  \begin{figure}[h]
  \centering
  \captionsetup{width=.9\linewidth}
  \includegraphics[width=0.85\linewidth]{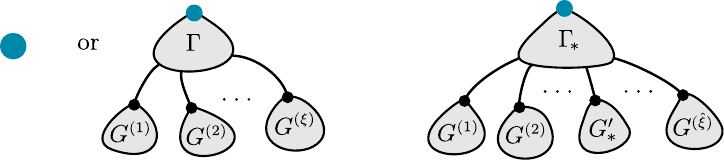}
  \caption{The recursive structure of $G$ (on the left) and of $G_*$ (on the right).
  The network $G$ is either reduced to a single vertex or obtained
  by grafting $\xi$ independent copies of itself on the leaves of
  a network $\Gamma \sim \nu_{\xi}$. By contrast, $G_*$ is never reduced
  to a single vertex: it is obtained by grafting $\hat{\xi}-1$
  independent copies of $G$  and an independent copy of itself on the leaves of
  a network $\Head{*} \sim \nu_{\hat{\xi}}$.}
  \label{figLocalLimitGalled}
\end{figure}

We now consider $G_n$ the subgraph of $G$ that corresponds to the blowup
of $[T]_n$, for each $n\geq 0$.
Note that $G_{n+1}$ is composed of finite graphs grafted on $G_n$, and that the
sequence $(G_n)_{n\geq0}$ increases to $G$: it is what we call a \emph{stub
sequence} of $G$, see Definition~\ref{defStubSeq} in the Appendix.
It is not hard to see -- and it is proved in Proposition~\ref{propStubSeq} --
that $(B_2(G_n))_{n\geq 0}$ is nondecreasing and tends to $B_2(G)$.
Let $\xi$ denote the number of children of the root of $T$, and write $\alpha_n
\defas \ExpecBrackets{B_2(G_n)}$ for conciseness. Conditional on $\Set{\xi = k}$,
we have that $q_i$ is independent of $G^{(i)}$ and that $\sum_{i=1}^k q_i = 1$.
Moreover, each network $G^{(i)}$ is distributed as $G$. As a result, by
taking expectations in Equation~\eqref{eqRecursionDistriB2} we get
\[
  \begin{dcases}
    \;\ExpecBrackets{B_2(G_{n+1}) \given \xi = 0} \;=\; 0 \\
    \;\ExpecBrackets{B_2(G_{n+1}) \given \xi = k} \;=\; \alpha_n + f(k) \,,
  \end{dcases}
\]
where $f(k) \defas \ExpecBrackets{B_2(\Gamma) \given \xi = k}$.
Plugging this in $\alpha_{n+1} = \sum_{k\geq 0} \eta_k \, \ExpecBrackets{B_2(G_{n+1}) \given \xi = k}$, we get
\[
 \alpha_{n+1} = \alpha_{n}(1-\eta_0) + \ExpecBrackets{f(\xi)}\,.
\]
Solving for $(\alpha_n)_{n\geq 0}$ yields
\[
  \alpha_0=0, \quad \alpha_1 = \ExpecBrackets{f(\xi)}, \quad \dots, \quad \alpha_n = \ExpecBrackets{f(\xi)} \sum_{k=0}^{n-1}(1-\eta_0)^{k} ,
\]
and by taking the limit as
$n\to\infty$, we get
$\alpha_n \to \alpha \defas \ExpecBrackets{B_2(G)} = \eta_0^{-1} \ExpecBrackets{f(\xi)}$,
proving the first part of the theorem.

Similarly, let $\smash{\hat{\xi}}$ denote the number of leaves of the root of
$T_*$, recalling that $\smash{\hat{\xi}}\sim \hat{\eta}$, where $\hat{\eta}$
denotes the size-biased distribution associated to $\eta$.
Let $\Head{*}$ denote the network associated to the root of $T_*$, and set
$\alpha_* \defas \ExpecBrackets{B_2(G_*)}$.  As previously, write $q_i$ for the
probability that the directed random walk started from the root of $\Head{*}$
goes through its leaf labeled $i$.
Finally, let $S$ denote the leaf of $\Head{*}$ through which the spine of $G_*$
goes, and write $I = \Set*{1, \ldots, \hat{\xi}} \setminus \Set{S}$.
Then,
\begin{equation}\label{eq:B_2(G^*)-recursive-eq}
    B_2(G_*) \;\mathrel{\overset{d}{=}}\;
  B_2(\Head{*}) \;+\; \sum_{i \in I} q_i\, B_2(G^{(i)}) \;+\;
  q_S\, B_2(G'_*)
\end{equation}
where $G^{(1)}, G^{(2)}, \ldots$ are independent copies of $G$, and
$G'_*$ is an independent copy of $G_*$.
Similarly, we define $(G_{*,n})_{n\geq 0}$ a stub sequence of
$G_*$ by removing the part of the graph descending from the $n$-th spinal
vertex.
Taking expectations and recalling
that $q_1 + \cdots + q_{\hat{\xi}} = 1$, we get
\begin{equation} \label{eqProofExpectedValueB2Blowup02}
  \beta_{n+1}\;\defas\;\ExpecBrackets{B_2(G_{*,n+1})}
    \;=\; \ExpecBrackets{B_2(\Head{*})} \;+\; \alpha \;+\; 
    (\beta_{n} - \alpha)\, \ExpecBrackets{q_S} \,, 
\end{equation}
In this expression, $\ExpecBrackets{B_2(\Head{*})} = \ExpecBrackets{f(\smash{\hat{\xi}}\,)}$.
Next, to compute $\ExpecBrackets{q_S}$, note that
\[
  \ExpecBrackets{q_S \given \Head{*}} \;=\;
  \sum_{i = 1}^{\hat{\xi}} q_i\, \underbrace{\Prob{S = i \given \Head{*}}}_{1/\hat{\xi}} \;=\;
  \frac{1}{\hat{\xi}} \,,
\]
from which it follows that $\ExpecBrackets{q_S} = \sum_{k \geq 1} \frac{\hat{\eta}_k}{k} =
\sum_{k \geq 1} \eta_k = 1 - \eta_0 $. Plugging this in
Equation~\eqref{eqProofExpectedValueB2Blowup02} and solving for $(\beta_{n})_{n\geq 0}$
then yields the desired expression for $\ExpecBrackets{B_2(G_{*})} = \lim_{n\to\infty} \beta_{n}$.
\end{proof}

Note that since the function $f$ in Theorem~\ref{thmExpecB2BlowUpGW} satisfies
$f(k) \leq \log_2 k$ for $k \geq 1$, this implies that if $G$ is a blowup
of~$T$, then $\ExpecBrackets{B_2(G)} \leq \ExpecBrackets{B_2(T)}$. In fact, it is not too hard
to see that we have the stronger statement
\[
  \ExpecBrackets{B_2(G) \given T} \;\leq\; B_2(T) \,, 
\]
as proved in Corollary~\ref{corIneqB2Blowup} from
Appendix~\ref{appProofsBlowUps}. In particular, this immediately implies that
if $T$ is a random tree (not necessarily a Galton--Watson tree) and $G$ is a
blowup of $T$, then $B_2(T)$ is second-order stochastically dominant over
$B_2(G)$; see Corollary~\ref{corStochasticDominance}.

Having introduced all the required tools, we now turn to
the study of the $B_2$ index of uniform leaf-labeled galled trees.

\subsection{Application to leaf-labeled galled trees} \label{secAppLocLimGalledTrees}

Let us start by briefly recalling the blowup construction of uniform
leaf-labeled galled trees introduced by Stufler in~\cite{stufler2022branching};
we refer the reader to this article for a rigorous justification
of the sampling procedure outlined below.

Following \cite{stufler2022branching}, let a \emph{head of size $k$} be:
\begin{itemize}
  \item for $k = 2$, a galled tree with 2 labeled leaves;
  \item for $k \geq 3$, a galled tree with a single gall of length $k+1$ and $k$ labeled leaves.
\end{itemize}
Let then $a_k$ denote the number of heads of size $k$ -- so that
$a_1 = 0$, $a_2 = 3$, $a_3 = 9$; see Figure~\ref{figBlowUpGalledTrees}.
To count the number of heads of size $k \geq 3$: first, orient the gall and
choose the position of the reticulation in the gall {($k$~possibilities);} then,
label the $k$~leaves {($k!$~possibilities);} finally, forget the orientation of
the cycle, and observe that each head is obtained exactly twice. This reasoning
also works for $k = 2$, but in that case we have to add the network consisting
of a single cherry. Putting the pieces together, we get $a_1 = 0$ and, 
for $k \geq 2$,
\[
  a_k \;=\; \frac{k \cdot k!}{2} \;+\; \Indic{k = 2} \,.
\]
Next, let $\nu = (\nu_k)_{k\geq 2}$, where $\nu_k$ denotes the uniform
distribution on the heads of size~$k$. It is not too hard to see that if
$T_n$ is a simply generated tree with $n$ leaves (see
e.g.~\cite{janson2012simply}) whose weight sequence is given by $w_k = a_k / k!$,
then the blowup of $T_n$ with respect to $\nu$ is a uniform leaf-labeled
galled tree.

\begin{figure}[h]
  \centering
  \captionsetup{width=.8\linewidth}
  \includegraphics[width=0.7\linewidth]{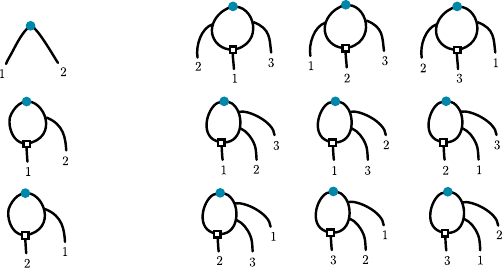}
  \caption{The heads of size $k$ used in the blowup construction of galled trees, for
  $k =2$ (left) and $k = 3$ (right). The blue dots are the roots, and the
  white squares are reticulations. Note that the heads are leaf-labeled.}
  \label{figBlowUpGalledTrees}
\end{figure}

To show that the simply generated tree $T_n$ can also be obtained as a
conditioned critical Galton--Watson tree, we must find an offspring
distribution $(\eta_k)_{k \geq 0}$ with mean~1 that is an exponential tilt of
the weight sequence $(w_k)_{k \geq 1}$ -- that is, we must find $\theta > 0$
such that
\[
  \begin{cases}
    \eta_k = \theta^{k-1} \, w_k\; \text{ for } k \geq 1 \text{ and }
       \eta_0 = 1 - \sum_{k \geq 1} \eta_k  & \text{(exponential tilt)} \\
    \sum_{k \geq 0} k\,\eta_k = 1 & \text{(criticality).}
  \end{cases}
\]
Note that the specific form of this exponential tilt (where, to get a probability
distribution, we change $\eta_0$ -- instead of by multiplying all $\eta_k$'s by
a constant) is due to the fact that we are conditioning on the number of
leaves (as opposed to  the number of vertices).

To find $\theta$, set $g(z) = \sum_{k\geq 1} w_k\, z^{k}$, so that
$\sum_{k\geq 1}k\,\eta_k = 1 \iff g'(\theta) = 1$. Elementary calculations
then show that
\[
  g'(z) \;=\; \frac{(2z^2 - 3z + 3)(z-2)z}{2(z-1)^3} \,,  
\]
from which we get
\[
  \theta \;=\; \frac{5 - \sqrt{17}}{4} \,.
\]
Note in passing that $\theta$ is also the first-order term in the expansion of
the generating function of galled trees at $z = 1/8$ -- see
Eq.~\eqref{exp-GTz}.

Let us recap the construction.

\begin{proposition} \label{propConstructionGalledTrees}
Let $T_n$ be a critical Galton--Watson tree conditioned to have $n$ leaves whose
offspring distribution $(\eta_k)_{k \geq 0}$ is defined by
\[
  \eta_0 = \frac{(2 - \theta)(\theta^2 - 3\theta  + 1)}{2(\theta - 1)^2},\quad\;
  \eta_1 = 0, \quad\;
  \eta_2 = \frac{3}{2}\, \theta, \quad\;
  \eta_k = \frac{k}{2}\,\theta^{k-1}\; (k \geq 3),
\]
with $\theta = (5 - \sqrt{17})/4$.
Let then $G_n$ be the blowup of $T_n$ with respect to $(\nu_k)_{k\geq 2}$, where
$\nu_k$ denotes the uniform distribution on the set of heads of size~$k$.
Then, $G_n$ is uniformly distributed on the set of galled trees with
$n$ labeled leaves.
\end{proposition}

Since the offspring distribution $\eta$ in
Proposition~\ref{propConstructionGalledTrees} has a finite third moment,
it follows from Theorem~\ref{thmContinuityB2BlowUps} that, as $n \to \infty$,
$B_2(G_n)$ converges in distribution and in all moments to $B_2(G_*)$, where
$G_*$ is the local limit of $G_n$ -- which, as we have already seen, is the
blowup with respect to $\nu$ of the Kesten tree $T_*$ associated to~$\eta$.
Therefore, $\lim_n \ExpecBrackets{B_2(G_n)} = \ExpecBrackets{B_2(G_*)}$ can be
computed using the expression given in Theorem~\ref{thmExpecB2BlowUpGW}.

Because when $\xi = 2$ the head can be a cherry (whose structure differs from
other heads, see Figure~\ref{figBlowUpGalledTrees}), it will be convenient to
let $\beta_k$ denote the expected $B_2$ index of a network $\Head{k}$
sampled uniformly at random among the galled trees with a single gall of length
$k+1$ and $k$ labeled leaves; and to write the function $f$ from
Theorem~\ref{thmExpecB2BlowUpGW} as $f(2) = \frac{1}{3} + \frac{2}{3} \beta_2$
for $k = 2$ and $f(k) = \beta_k$ for $k\geq 3$.

To compute $\beta_k$, orient the gall of $\Head{k}$ uniformly at
random, then number its vertices from $1$ to $k$. For $r = 0, \ldots, k-1$, let
$\Head{k,r}$ denote the realization of $\Head{k}$ whose reticulation is on the
{$(r+1)$-th} vertex of the gall -- that is, there are $r$ vertices on one side
of the gall and $k-r-1$ on the other. Letting $\beta_{k, r} = B_2(\Head{k, r})$,
we have
\[
  \beta_{k, r} \;=\;
  \sum_{i = 1}^r \frac{i+1}{2^{i+1}} \;+\,
  \sum_{j = 1}^{k-r-1} \frac{j+1}{2^{j+1}} \;-\;
  \left(\frac{1}{2^{r+1}} + \frac{1}{2^{k-r}}\right)
  \log_2 \left(\frac{1}{2^{r+1}} + \frac{1}{2^{k-r}}\right)\,.
\]
Next, we need to integrate over the position of the reticulation, as described
by $r$ in the parametrization above.
Using the procedure described at the beginning of the section to enumerate the
heads, we see that $r$ is uniform on $\Set{0, \ldots, k-1}$. Therefore,
\begin{equation} \label{eqBetaK}
  \beta_k \;=\; \frac{1}{k} \sum_{r = 0}^{k-1} \beta_{k, r}\,.
\end{equation}

Finally, by substituting the expression for $\eta_k$ from
Proposition~\ref{propConstructionGalledTrees} in the expression of
Theorem~\ref{thmExpecB2BlowUpGW}, we get --
in the notation of the theorem,
\[
  \ExpecBrackets{f(\xi)} \;=\;
  \frac{\theta}{2} \;+\; \frac{1}{2}\sum_{k\geq 2} k\, \theta^{k-1}\,\beta_k
  \quad\text{and}\quad
  \ExpecBrackets*{\normalsize}{f(\hat{\xi}\,)} \;=\;
  \theta \;+\; \frac{1}{2}\,\sum_{k \geq 2} k^2\, \theta^{k-1} \beta_k \,.
\]
The following theorem summarizes the calculations of this section;
compare with the results from Section~\ref{secAnalyticCombinatorics}.

\begin{theorem}
Let $G_n$ be sampled uniformly at random among the set of galled trees with
$n$ labeled leaves. Then, as $n \to \infty$, $B_2(G_n)$ converges in distribution
and in all moments to a random variable $B_2(G_*)$ with finite moments of all
orders. In particular,
\[
  \ExpecBrackets{B_2(G_*)} \;=\; 
  \frac{1}{2\,\eta_0} \mleft(3\, \theta \;+\; \sum_{k \geq 2} (k+1)k\, \theta^{k-1} \beta_k \mright) \,,
\]
where $\theta$ and $\eta_0$ are explicit algebraic constants
given in Proposition~\ref{propConstructionGalledTrees}, and $\beta_k$ is
given in Equation~\eqref{eqBetaK}. Numerically,
$\alpha_* = 2.707911858984\dots$.
\end{theorem}

\begin{remark} \label{rem:higher-moments-computations-proba}
The recursive structure of $G^*$ could in principle be used to compute
the higher moments of $B_2(G^*)$.
Indeed, raising both sides of Eq.~\eqref{eq:B_2(G^*)-recursive-eq} to the power~$\ell$ and
taking expectations, we get a recursion for
$\ExpecBrackets*{\normalsize}{B_2(G^*)^\ell}$ involving
$\ExpecBrackets*{\normalsize}{B_2(G)^\ell}$ and all lower-order moments.
However, the computations quickly become tedious and we did not push this approach further.
\end{remark}

\section*{Acknowledgments}

Part of this work was done while FB was at the Institute for Theoretical
Studies of ETH Zürich and was supported by Dr.~Max Rössler, the Walter Haefner
Foundation and the ETH Zürich Foundation. 
FB and JJD were in part supported by Agence Nationale de la Recherche (ANR),
grant ANR-24-CE40-7154.
MF and TCY were in part supported by National Science and Technology Council
(NSTC), Taiwan under grant NSTC-111-2115-M-004-002-MY2. 

%\pagebreak
\addcontentsline{toc}{section}{References}
\bibliographystyle{abbrvnat}
\bibliography{refs}

\appendix

\newpage 

\section*{Appendices} \label{secAppendices}
\addcontentsline{toc}{section}{Appendices}
\refstepcounter{section} % Otherwise sections are numbered .x instead of A.x

\subsection{The limit law via the method of moments} \label{ll-mom}

The main purpose of this appendix is to prove the technical estimate (\ref{upper-bound-dl}) (see Lemma~\ref{bound-dl} below) in order to complete the proof of Theorem~\ref{ll-B2}. Set 
\begin{align*}
f_{\ell,1}(z)&\defas \frac{1}{2}\sum_{\substack{\ell_1+\ell_2+\ell_3=\ell\\\ell_1,\ell_2<\ell}}\binom{\ell}{\ell_1,\ell_2,\ell_3}\frac{1}{2^{\ell_1+\ell_2}}A^{[\ell_1]}(z)A^{[\ell_2]}(z),\\
f_{\ell,2}(z)&\defas\sum_{k\geq 1}\sum_{\substack{\ell_1+\cdots+\ell_{k+2}=\ell\\\ell_j<\ell, 1\leq j\leq k+1}}\binom{\ell}{\ell_1,\ldots,\ell_{k+2}}\prod_{s=1}^{k}\frac{1}{2^{(s+1)\ell_s}}\hspace{-0.2em}\left(\frac{1}{2}+\frac{1}{2^{k+1}}\right)^{\ell_{k+1}}\hspace{-0.7em}\mu_{k,0}^{\ell_{k+2}}\\
&\qquad\quad\times\prod_{s=1}^{k+1}A^{[\ell_s]}(z),\\
f_{\ell,3}(z)&\defas\frac{1}{2}\sum_{k,m\geq 1}\sum_{\substack{\ell_1+\cdots+\ell_{k+m+2}=\ell\\\ell_j<\ell,1\leq j\leq k+m+1}}\binom{\ell}{\ell_1,\cdots,\ell_{k+m+2}}\prod_{s=1}^{k}\frac{1}{2^{(s+1)\ell_s}}\prod_{t=1}^{m}\frac{1}{2^{(t+1)\ell_{k+t}}}\nonumber\\
&\qquad\quad\times\left(\frac{1}{2^{k+1}}+\frac{1}{2^{m+1}}\right)^{\ell_{k+m+1}}\mu_{k,m}^{\ell_{k+m+2}}\prod_{s=1}^{k}A^{[\ell_s]}(z)\prod_{t=1}^{m+1}A^{[\ell_{k+t}]}(z)
\end{align*}
so that $f_{\ell}(z)=f_{\ell,1}(z)+f_{\ell,2}(z)+f_{\ell,3}(z)$; see (\ref{flz}), where we have split the last sum of (\ref{flz}) into two parts since below, we use the first (simpler) part to explain our ideas and then treat the second (more complicated) part without repeating the similar details. Moreover, set 
\[
c_{\ell,i}\defas f_{\ell,i}(A^{[j]}(z)\leftrightarrow c_j),\qquad
d_{\ell,i}\defas f'_{\ell,i}((A^{[j]})'(z)\leftrightarrow d_j,A^{[j]}(z)\leftrightarrow c_j),
\]
and thus
\begin{equation}\label{cl+dl}
c_{\ell}=\frac{c_{\ell,1}+c_{\ell,2}+c_{\ell,3}}{g_{\ell}(\rho)},\qquad d_{\ell}=\frac{d_{\ell,1}+d_{\ell,2}+d_{\ell,3}}{g_{\ell}(\rho)}-c_{\ell}\frac{g_{\ell}'(\rho)}{g_{\ell}(\rho)},
\end{equation}
where notation is as in Section~\ref{secAnalyticCombinatorics}.

We first consider $c_{\ell}$ for which we have the following result.
\begin{lemma}\label{bound-cl}
There exists a constant $K>0$ such that $c_\ell\leq K^\ell\ell!$ for all $\ell\geq 0$.
\end{lemma}
\begin{proof}
We prove the claim by induction on $\ell$, where we take
\[
  K \defas\max\left\{c_1,\sqrt{\frac{c_2}{2}},\ldots,\sqrt[9]{\frac{c_9}{9!}},512\right\}.
\]
Note that by this choice, the claim holds for $1\leq\ell\leq 9$. In addition, it trivially holds for $\ell=0$. Also, note that $g_{\ell}(\rho)$ is increasing in $\ell$ and $g_1(\rho)\equiv 0.570194101601\cdots$. Thus, by (\ref{cl+dl}), we have to show that
\[
c_{\ell,1}+c_{\ell,2}+c_{\ell,3}\leq g_{1}(\rho)K^{\ell}\ell!
\]
for $\ell\geq 10$. We assume that this claim holds for $\ell'<\ell$ and prove it for $\ell$. For this, we estimate $c_{\ell,i}$ with $i=1,2,3$ separately. 

First, for $c_{\ell,1}$, we have:
\begin{align}
c_{\ell,1} & \leq \frac{1}{2}\sum_{\substack{\ell_1+\ell_2+\ell_3=\ell\\ \ell_1,\ell_2<\ell}}\frac{\ell!}{\ell_3!}\frac{1}{2^{\ell_1+\ell_2}}K^{\ell_1+\ell_2}\leq\frac{\ell!}{2}\sum_{\substack{\ell_1+\ell_2+\ell_3=\ell\\ \ell_1,\ell_2<\ell}}\frac{K^{\ell-\ell_3}}{2^{\ell-\ell_3}}\nonumber\\
&=\frac{\ell!}{2}\left(\frac{K}{2}\right)^{\ell}\sum_{\substack{\ell_1+\ell_2+\ell_3=\ell\\ \ell_1,\ell_2<\ell}}\left(\frac{2}{K}\right)^{\ell_3}\leq\frac{\ell!}{2}\left(\frac{K}{2}\right)^\ell\sum_{\ell_3=0}^{\ell}\left(\frac{1}{256}\right)^{\ell_3}\sum_{\substack{\ell_1+\ell_2=\ell-\ell_3\\\ell_1,\ell_2<\ell}}1\nonumber\\
&\leq K^{\ell}\ell!\left(\frac{1}{2}\right)^{\ell+1}\sum_{\ell_3=0}^{\ell}\left(\frac{1}{256}\right)^{\ell_3}(\ell-\ell_3+1)\nonumber\\
&=K^{\ell}\ell!\left(\frac{1}{2}\right)^{\ell+1}\left(\frac{256}{255}\ell+\frac{256^{-\ell}}{65025}+\frac{65024}{65025}\right)\nonumber\\
 & \leq 0.005390234525\cdots K^{\ell}\ell!\label{c1},
\end{align}
where in the estimate of the second line, we used that $K\geq 512$, and the last inequality holds for $\ell\geq 10$.

Next, we consider $c_{\ell,2}$. Here, we first note that since $\ell_1+\cdots+\ell_{k+2}=\ell$, we have $k+1$ degrees of freedom, i.e., $\ell_{k+1}$ is fixed if $\ell_{1},\ldots,\ell_{k}$ and $\ell_{k+2}$ are decided. We replace $c_{\ell_{k+1}}$ on the right-hand side of $c_{\ell,2}$ by $K^{\ell_{k+1}}\ell_{k+1}!$ except when $\ell_{k+1}=0$ and $\ell_{k+1}=1$ where we replace it by $\rho$ and $c_1$, respectively. This gives
\begin{equation}\label{cl2-est}
c_{\ell,2}\leq\frac{9}{16}\sum_{k\geq 1}\sideset{}{'}\sum\binom{\ell}{\ell_1,\ldots,\ell_{k+2}}\prod_{s=1}^{k}\frac{1}{2^{(s+1)\ell_s}}\mu_{k,0}^{\ell_{k+2}}\left(\prod_{s=1}^{k} c_{\ell_s}\right)K^{\ell_{k+1}}\ell_{k+1}!,
\end{equation}
where $\ell_{k+1}=\ell-\ell_1-\cdots-\ell_{k}-\ell_{k+2}$, the second sum is over all $0<\ell_1+\cdots+\ell_{k}+\ell_{k+2}\leq\ell$ with $\ell_j<\ell,1\leq j\leq k$, and we have used that $\xi_{k,\ell_{k+1}}\leq9/16$ with
\[
\xi_{k,\ell_{k+1}}=\begin{cases}\rho,&\text{if}\ \ell_{k+1}=0,\\
\left(\frac{1}{2}+\frac{1}{2^{k+1}}\right)\frac{c_1}{K},&\text{if}\ \ell_{k+1}=1,\\
\left(\frac{1}{2}+\frac{1}{2^{k+1}}\right)^{\ell_{k+1}},&\text{if}\ \ell_{k+1}\geq 2.\end{cases}
\]
Next, by using $\mu_{k,0}\leq 2$ and plugging the induction hypothesis for $c_{\ell_s}$ into (\ref{cl2-est}) except when $s\neq k+1$ where we again use $\rho$ and $c_1$ if $\ell_s=0$ and $\ell_s=1$, we have
\begin{align}
c_{\ell,2}&\leq\frac{9}{16}K^{\ell}\ell!\sum_{k\geq 1}\prod_{s=1}^{k}\left(\rho+\frac{c_1}{K2^{s+1}}+\sum_{\ell_s\geq 2}\frac{1}{2^{(s+1)\ell_s}}\right)\sum_{\ell_{k+2}=0}^{\infty}\frac{2^{\ell_{k+2}}}{K^{\ell_{k+2}}\ell_{k+2}!}\nonumber\\
&\leq\frac{9e^{1/256}}{16}K^{\ell}\ell!\sum_{k\geq 1}\prod_{s=1}^{k}\left(\rho+\frac{c_1}{K2^{s+1}}+\frac{1}{2^{s+1}(2^{s+1}-1)}\right)\nonumber\\
&=0.223033369804\cdots K^{\ell}\ell!,\label{c2}
\end{align}
where we used that $K\geq 512$ in the second last step and Maple to evaluate the numerical constant.

Similar, we obtain for $c_{\ell,3}$ by using $\mu_{k,m}\leq 3$ that
\begin{align}
c_{\ell,3}&\leq\frac{\rho e^{3/512}}{2}K^{\ell}\ell!\sum_{k\geq 1}\prod_{s=1}^{k}\left(\rho+\frac{c_1}{K2^{s+1}}+\frac{1}{2^{s+1}(2^{s+1}-1)}\right)\nonumber\\
&\qquad\qquad\qquad\qquad\times\sum_{m\geq 1}\prod_{t=1}^{m}\left(\rho+\frac{c_1}{K2^{t+1}}+\frac{1}{2^{t+1}(2^{t+1}-1)}\right)\nonumber\\
&=0.017199020110\cdots K^{\ell}\ell!.\label{c3}
\end{align}

Combining (\ref{c1}), (\ref{c2}), and (\ref{c3}) gives
\[
c_{\ell,1}+c_{\ell,2}+c_{\ell,3}\leq 0.245622624440\cdots K^{\ell}\ell!\leq g_1(\rho)K^{\ell}\ell!
\]
which proves the desired result.
\end{proof}

We use this now to prove a similar result for $d_{\ell}$.

\begin{lemma}\label{bound-dl}
There exists a constant $K>0$ such that $d_{\ell}\leq K^{\ell}\ell!$ for all $\ell\geq 0$.
\end{lemma}
\begin{proof}
We again use induction on $\ell$ where the claim holds for $\ell=0$. Moreover, we choose $K$ large enough such that the conclusion of Lemma~\ref{bound-cl} holds and in addition that 
\begin{itemize}
\item[(i)] $K\geq\max\left\{d_1,\sqrt{d_2/2},\ldots,\sqrt{d_9/9!},4608c'/c_1,4608d'/d_1\right\}$;
\item[(ii)] $\rho\geq\max\left\{3c_1/(4K),3^2c_2/(2! 4^2K^2),\ldots,3^{9}c_9/(9!4^9K^9)\right\}$,
\end{itemize}
where $c'\defas \max\{c_1,\ldots,c_9\}$ and $d'\defas \max\{d_1,\ldots,d_9\}$. 

By the item (i), the claim holds for all $1\leq\ell\leq 9$. Thus, we can assume that $\ell\geq 10$ and that the claim holds for all $\ell'<\ell$. We are going to bound $d_{\ell,i}$ for $i=1,2,3$.

First for $d_{\ell,1}$, we have
\begin{align*}
f'_{\ell,1}(z)&=\frac{1}{2}\sum_{\substack{\ell_1+\ell_2+\ell_3=\ell\\ \ell_1,\ell_2<\ell}}\binom{\ell}{\ell_1,\ell_2,\ell_3}\frac{1}{2^{\ell_1+\ell_2}}\left((A^{[\ell_1]})'(z)A^{[\ell_2]}(z)+A^{[\ell_1]}(z)(A^{[\ell_2]})'(z)\right)
\end{align*}
and thus
\begin{align*}
d_{\ell,1} & = \frac{1}{2}\sum_{\substack{\ell_1+\ell_2+\ell_3=\ell\\ \ell_1,\ell_2<\ell}}\binom{\ell}{\ell_1,\ell_2, \ell_3}\frac{1}{2^{\ell_1+\ell_2}}\left(d_{\ell_1}c_{\ell_2}+c_{\ell_1}d_{\ell_2}\right)\\
&=\sum_{\substack{\ell_1+\ell_2+\ell_3=\ell\\ \ell_1,\ell_2<\ell}}\binom{\ell}{\ell_1,\ell_2, \ell_3}\frac{1}{2^{\ell_1+\ell_2}}c_{\ell_1}d_{\ell_2}.
\end{align*}
Using the induction hypothesis and the estimate from Lemma~\ref{bound-cl}, we obtain
\begin{equation}\label{d1}
d_{\ell,1}\leq\sum_{\substack{\ell_1+\ell_2+\ell_3=\ell\\ \ell_1,\ell_2<\ell}}\frac{\ell!}{\ell_3!}\frac{1}{2^{\ell_1+\ell_2}}K^{\ell_1+\ell_2}\leq0.010780469050\cdots K^{\ell}\ell!
\end{equation}
which is twice the bound from (\ref{c1}).

Next, for $d_{\ell,2}$, we have
\begin{align*}
f'_{\ell,2}(z)=\sum_{k\geq 1}\sum_{\substack{\ell_1+\cdots+\ell_{k+2}=\ell\\ \ell_j<\ell,1\leq j\leq k+1}}&\binom{\ell}{\ell_1,\cdots,\ell_{k+2}}\prod_{s=1}^k\frac{1}{2^{(s+1)\ell_s}}\left(\frac{1}{2}+\frac{1}{2^{k+1}}\right)^{\ell_{k+1}}\mu_{k,0}^{\ell_{k+2}}\\
&\times\left(\sum_{i=1}^{k+1}(A^{[\ell_i]})'(z)\prod_{\substack{1\leq j\leq k+1\\j\neq i}}A^{[\ell_j]}(z)\right)
\end{align*}
and thus
\begin{align*}
d_{\ell,2}&=\sum_{k\geq 1}\sum_{\substack{\ell_1+\cdots+\ell_{k+2}=\ell\\ \ell_j<\ell,1\leq j\leq k+1}}\binom{\ell}{\ell_1,\cdots,\ell_{k+2}}\prod_{s=1}^k\frac{1}{2^{(s+1)\ell_s}}\left(\frac{1}{2}+\frac{1}{2^{k+1}}\right)^{\ell_{k+1}}\mu_{k,0}^{\ell_{k+2}}\\
&\qquad\qquad\qquad\times\left(\sum_{i=1}^{k+1}d_{\ell_i}\prod_{\substack{j=1\\j\neq i}}^{k+1}c_{\ell_j}\right)\\
&=\sum_{k\geq 1}\sum_{\substack{\ell_1+\cdots+\ell_{k+2}=\ell\\ \ell_j<\ell,1\leq j\leq k+1}}\binom{\ell}{\ell_1,\cdots,\ell_{k+2}}\prod_{s=1}^k\frac{1}{2^{(s+1)\ell_s}}\left(\frac{1}{2}+\frac{1}{2^{k+1}}\right)^{\ell_{k+1}}\mu_{k,0}^{\ell_{k+2}}\\
&\qquad\qquad\qquad\times\left(\sum_{i=1}^{k}d_{\ell_i}\prod_{\substack{j=1\\j\neq i}}^{k+1}c_{\ell_j}\right)\\
&\quad+\sum_{k\geq 1}\sum_{\substack{\ell_1+\cdots+\ell_{k+2}=\ell\\ \ell_j<\ell,1\leq j\leq k+1}}\binom{\ell}{\ell_1,\cdots,\ell_{k+2}}\prod_{s=1}^k\frac{1}{2^{(s+1)\ell_s}}\left(\frac{1}{2}+\frac{1}{2^{k+1}}\right)^{\ell_{k+1}}\mu_{k,0}^{\ell_{k+2}}\\
&\qquad\qquad\qquad\times d_{\ell_{k+1}}\prod_{j=1}^{k}c_{\ell_j}.
\end{align*}
We use now both the estimate for $c_{\ell}$ from Lemma~\ref{bound-cl} and the induction hypothesis for $d_{\ell}$ except for $0\leq\ell\leq 9$. Moreover, we use $\mu_{k,0}\leq 2$ (as in the proof of Lemma~\ref{bound-cl}) and
\[
\left(\frac{1}{2}+\frac{1}{2^{k+1}}\right)^{\ell_{k+1}}c_{\ell_{k+1}}\leq\left(\frac{3}{4}\right)^{\ell_{k+1}}c_{\ell_{k+1}}\leq\rho K^{\ell_{k+1}}\ell_{k+1}!
\]
which follows by item (ii) above and 
\[
\left(\frac{1}{2}+\frac{1}{2^{k+1}}\right)^{\ell_{k+1}}d_{\ell_{k+1}}\leq K^{\ell_{k+1}}\ell_{k+1}!
\]
which holds trivially. The rest is handled by similar ideas as in the proof of Lemma~\ref{bound-cl}. Thus, we obtain
\begin{align*}
d_{\ell,2}&\leq\rho e^{2/K} K^{\ell}\ell!\sum_{k\geq 1}\prod_{s=1}^k\left(\rho+\left(\sum_{i=1}^9\frac{c_i}{2^{i(s+1)}i!K^i}\right)+\frac{1}{2^{9(s+1)}(2^{s+1}-1)}\right)\\
&\qquad\qquad\qquad\times\left(\sum_{s=1}^k\frac{1+\left(\sum_{i=1}^9\frac{d_i}{2^{i(s+1)}i!K^i}\right)+\frac{1}{2^{9(s+1)}(2^{s+1}-1)}}{\rho+\left(\sum_{i=1}^9\frac{c_i}{2^{i(s+1)}i!K^i}\right)+\frac{1}{2^{9(s+1)}(2^{s+1}-1)}}\right)\\
&\hspace{1.25 em}+ e^{2/K} K^{\ell}\ell!\sum_{k\geq 1}\prod_{s=1}^k\left(\rho+\left(\sum_{i=1}^9\frac{c_i}{2^{i(s+1)}i!K^i}\right)+\frac{1}{2^{9(s+1)}(2^{s+1}-1)}\right).
\end{align*}
Note that, by our choice of $K$,
\[
\sum_{i=1}^{9}\frac{c_i}{2^{i(s+1)}i!K^i}\leq\frac{9c'}{2^{s+1}K}\leq\frac{1}{512\cdot 2^{s+1}}
\]
and likewise with the $c_i$'s replaced by $d_i$'s. Plugging this into the above expression and numerical evaluating it (again with Maple), we obtain
\begin{equation}\label{d2}
d_{\ell,2}\leq0.643482769458\cdots K^{\ell}\ell!.
\end{equation}

We finally consider $d_{\ell,3}$ which is treated in a similar fashion. First,
\begin{align*}
f'_{\ell,3}(z)&=\frac{1}{2}\sum_{k\geq 1}\sum_{m\geq 1}\sum_{\substack{\ell_1+\cdots+\ell_{k+m+2}=\ell \\ \ell_j< \ell,1\leq j\leq k+m+1}}\binom{\ell}{\ell_1, \cdots,\ell_{k+m+2}}\prod_{s=1}^k\frac{1}{2^{(s+1)\ell_s}}\prod_{t=1}^m\frac{1}{2^{(t+1)\ell_{k+t}}}\\
&\qquad\times\left(\frac{1}{2^{k+1}}+\frac{1}{2^{m+1}}\right)^{\ell_{k+m+1}}\mu_{k,m}^{\ell_{k+m+2}}\left(\sum_{i=1}^{k+m+1}(A^{[\ell_i]})'(z)\prod_{\substack{j=1\\j\neq i}}^{k+m+1}
A^{[\ell_j]}(z)\right)
\end{align*}
and thus
\begin{align*}
d_{\ell,3}&= \frac{1}{2}\sum_{k\geq 1}\sum_{m\geq 1}\sum_{\substack{\ell_1+\cdots+\ell_{k+m+2}=\ell \\ \ell_j<\ell,1\leq j\leq k+m+1}}\binom{\ell}{\ell_1, \cdots,\ell_{k+m+2}}\prod_{s=1}^k\frac{1}{2^{(s+1)\ell_s}}\prod_{t=1}^m\frac{1}{2^{(t+1)\ell_{k+t}}}\\
&\qquad\quad\times\left(\frac{1}{2^{k+1}}+\frac{1}{2^{m+1}}\right)^{\ell_{k+m+1}}\mu_{k,m}^{\ell_{k+m+2}}\left(\sum_{i=1}^{k+m}d_{\ell_i}\prod_{\substack{j=1\\j\neq i}}^{k+m+1}
c_{\ell_j}\right)\\
& \hspace{1 em} +\frac{1}{2}\sum_{k\geq 1}\sum_{m\geq 1}\sum_{\substack{\ell_1+\cdots+\ell_{k+m+2}=\ell \\ \ell_j<\ell,1\leq j\leq k+m+1}}\binom{\ell}{\ell_1, \cdots, \ell_{k+m+2}}\prod_{s=1}^k\frac{1}{2^{(s+1)\ell_s}}\prod_{t=1}^m\frac{1}{2^{(t+1)\ell_{k+t}}}\\
&\qquad\quad\times\left(\frac{1}{2^{k+1}}+\frac{1}{2^{m+1}}\right)^{\ell_{k+m+1}}\mu_{k,m}^{\ell_{k+m+2}}d_{\ell_{k+m+1}}\prod_{i=1}^{k+m}
c_{\ell_i}(z).
\end{align*}
Now, using $\mu_{k,m}\leq 3$ (see the proof of Lemma~\ref{bound-cl}) and arguments as above, we obtain
\begin{align}
&d_{\ell,3}\leq\rho e^{3/K} K^\ell \ell!\sum_{k\geq 1}\prod_{s=1}^k\left(\rho+\frac{c_1}{512\cdot2^{(s+1)}K}+\frac{1}{2^{9(s+1)}(2^{s+1}-1)}\right)\nonumber\\
&\times \hspace{-0.1cm}\left(\sum_{s=1}^k\frac{1+\frac{d_1}{512\cdot2^{(s+1)}}+\frac{1}{2^{9(s+1)}(2^{s+1}-1)}}{\rho+\frac{c_1}{512\cdot2^{(s+1)}}+\frac{1}{2^{9(s+1)}(2^{s+1}-1)}}\right)\sum_{k\geq 1}\prod_{s=1}^k\left(\rho+\frac{c_1}{512\cdot2^{(s+1)}}+\frac{1}{2^{9(s+1)}(2^{s+1}-1)}\right)\nonumber\\
&\ +\frac{e^{3/K}K^\ell \ell!}{2}\left(\sum_{k\geq 1}\prod_{s=1}^k\left(\rho+\frac{c_1}{512\cdot2^{(s+1)}}+\frac{1}{2^{9(s+1)}(2^{s+1}-1)}\right)\right)^2\nonumber\\
&\ =0.141422204463\cdots K^{\ell}\ell!\label{d3}.
\end{align}

Combining (\ref{d1}),(\ref{d2}), and (\ref{d3}) gives:
\[
d_{\ell,1}+d_{\ell,2}+d_{\ell,3}\leq 0.795685442972\cdots K^{\ell}\ell!.
\]
Finally, for $\ell\geq 10$,
\begin{align*}
d_{\ell}&\leq\frac{d_{\ell,1}+d_{\ell,2}+d_{\ell,3}}{g_{10}(\rho)}-\frac{g_{10}'(\rho)}{g_{10}(\rho)}c_{\ell}\\
&\leq0.869090272578\cdots K^{\ell}\ell!
\end{align*}
since $-g_{\ell}'(\rho)/g_{\ell}(\rho)$ decreases to $0$ and $g_{10}(\rho)= 0.986940779096\cdots$ and $-g_{10}'(\rho)/g_{10}(\rho)=0.062876303286\cdots.$ This proves the claim.
\end{proof}

%\pagebreak

\subsection{The boundary of a phylogenetic network} \label{appBoundary}

Recall from Section~\ref{secB2InfiniteNetwork} that, to extend
the definition of $B_2$ to infinite phylogenetic networks, we need
a notion of ``boundary'' of a network that includes not only the
vertices in which the directed random walk $X$ can get trapped (i.e.\ the
leaves), but also all of the distinct ways in which it can escape to infinity
(i.e.\ the ends, see Definition~\ref{defEnds}). This motivated
Definition~\ref{defBoundary} of the boundary $\partial G$ of a phylogenetic
network $G$ as $\partial G = \mathcal{L} \cup \mathcal{E}$, where $\mathcal{L}$
and $\mathcal{E}$ denote the set of leaves and the set of ends of~$G$,
respectively.

In this appendix, we show that $\partial G$ can be embedded in a suitable
compact metric space $(\mathcal{K},\, d_{\mathcal{K}})$, and that the escape
point $X_\infty$ of the directed random walk $X$ is a well-defined random variable
in~$\mathcal{K}$. We also discuss some topological properties of
the boundary $\partial G$, in particular its connection with the Martin
boundary.

Let us start by introducing some notation. We write $u\ancestor v$ to indicate
that $u$ is an ancestor of $v$, i.e.\ that there exists a finite directed path
from $u$ to $v$ in $G$. For any vertex $v \in G$, we denote by
\[
  \bar{v} \;=\, \Set{u\in G \suchthat u \ancestor v} 
\]
the set of ancestors of $v$. Finally, for any set of vertices $S \subset G$, we
use the short notation $[S]_k$ for the set $S \cap [G]_k$ of vertices of $S$ at
height at most $k$ in $G$.

\begin{definition} \label{defK}
Let $G$ be a phylogenetic network. Set
\[
  \mathcal{K} \;=\; \Set{S\subset G \suchthat \forall u\in S,\, \bar{u}\subset S}\,,
\]
and equip $\mathcal{K}$ with the distance $d_{\mathcal{K}}$ defined by
\[
  d_{\mathcal{K}}(S, S') \;=\;
  \sup\Set{2^{-n} \suchthat n \in \mathbb{N},\, [S]_n \neq [S']_n},
\]
with the convention $\sup \emptyset = 0$.
\end{definition}

It is readily checked that $(\mathcal{K}, d_{\mathcal{K}})$ is a compact
metric space (using a diagonal argument and the local finiteness of $G$
to show compactness).
% We never use the fact that it is *ultra*metric, so I removed "ultra".

Because the map $v\in G\mapsto \bar{v}\in \mathcal{K}$ is
injective, the vertices of $G$ can be seen as points of $\mathcal{K}$.
The next proposition shows that the ends of $G$ can also be
embedded in $\mathcal{K}$; more precisely, they are points of $\mathrm{cl}(G)$,
the topological closure of $G$ in~$\mathcal{K}$.

\begin{proposition} \label{propK}
Let $G$ be a phylogenetic network, and let $\mathcal{E}$ denote its end set.
For each $x \in \mathcal{E}$, define
\[
  \bar{x} \;=\; \Set{v\in G \suchthat \exists r\in x \text{ such that }v\in r},
\]
where $v\in r$ means that the ray~$r$ goes through the vertex~$v$. Then, the map
\[
  x\in \mathcal{E} \;\longmapsto\; \bar{x} \in \mathcal{K}
\]
is injective. Furthermore, for any ray
$r = (v_n)_{n \geq 0}$, we have $\bar{v}_n \to\bar{x}$ in
$(\mathcal{K}, d_\mathcal{K})$, where
$x$ is the end associated with $r$.
\end{proposition}

\begin{proof}
First, let us show that $\bar{x}\in \mathcal{K}$ for any $x\in \mathcal{E}$.
Fix $x \in \mathcal{E}$, $v \in \bar{x}$ and $u \ancestor v$.
We must show that $u \in \bar{x}$, i.e.\ that $x$ contains a ray that goes
through $u$. Let $r\in x$ be such that $v \in r$. Because $G$ is rooted and
$u \ancestor v$, there exists a path that goes from the root of $G$ to~$u$, and
then from $u$ to~$v$. Let $r'$ be the ray that starts with this path and then
continues along~$r$. Since $r' \codir r$, we have $r' \in x$, concluding
the proof that $\bar{x} \in \mathcal{K}$.

Let us now show that $x\mapsto \bar{x}$ is injective. Consider
$x, y \in \mathcal{E}$ such that $\bar{x}=\bar{y}$, and fix two rays
$r = (v_i)_{i\geq 0}\in x$, and $r' = (v'_i)_{i\geq 0}\in y$.  To show that
$x = y$, we must show that $r \codir r'$, and for this it suffices to exhibit a
ray that intersects $r$ and $r'$ infinitely many times.  The key is to show
that for all $i\geq 0$, there exists $j \geq 0$ such that $v_i \ancestor v'_j$
-- indeed, by symmetry we then get that there exists
$k \geq 0$ such that $v'_j \ancestor v_k$, etc; and by a straightforward concatenation
procedure we can build a ray that goes to and fro between vertices of $r$ and
vertices of $r'$. By assumption, $\bar{x} = \bar{y}$ and thus
$v_i \in \bar{y}$. As a result, there exists a ray $r''\codir r'$ such that
$v_i \in r''$, and since $r''$ intersects $r'$ infinitely many times, there
exists $u \in r''$, $u \succcurlyeq v_i$ such that $u \in r'$, concluding
the proof of the injectivity of $x\mapsto \bar{x}$.

Finally, consider any ray $r = (v_n)_{n\geq 0}$, letting $x$ denote the 
corresponding end.  Let us show that $\bar{v}_n \to \bar{x}$ as $n \to \infty$,
i.e.\ let us fix some $k\geq 0$ and show that $[\bar{x}]_k = [\bar{v}_i]_k$
for all $i$ large enough.  Let $v\in [\bar{x}]_k$, and pick $r_v\in x$ such
that $v\in r_v$.  Because $r_v\codir r$, there exists $i_v\in \mathbb{N}$ such that
$v\ancestor v_{i_v}$. Since $[G]_k$ is finite, we can define
\[
  N_k \;=\; \max \Set{i_v \suchthat v\in [\bar{x}]_k}\,.
\]
Thus, $[\bar{x}]_k \subset \bar{v}_i$ for all $i\geq N_k$. Moreover, because
$\bar{x}\in \mathcal{K}$ we also have $\bar{v}_i\subset \bar{x}$ for
all~$i\geq 0$. As a result, $[\bar{x}]_k = [\bar{v}_i]_k$ for all $i\geq N_k$,
concluding the proof.
\end{proof}

Proposition~\ref{propK} shows that the vertices and the ends of $G$ can be
embedded in $\mathcal{K}$ (in the rest of this section, we thus identify $v\in G$ with $\bar{v}\in \mathcal{K}$ and simply drop the notation $\bar{v}$), and it immediately implies the following corollary.

\begin{corollary} \label{corXinftyWellDefined}
Let $G$ be a phylogenetic network and let $(X_t)_{t \geq 0}$ be the directed
random walk on $G$. Then, $X_t$ converges almost surely in
$(\mathcal{K},\, d_{\mathcal{K}})$ to a random variable $X_\infty \in \partial G$.
\end{corollary}

\begin{remark}
Note that there is no connection between the boundary $\partial G$ and the
topological boundary of $G$ in $\mathcal{K}$, that is,
$\mathrm{cl}(G) \setminus \mathrm{int}(G)$. Indeed:
\begin{mathlist}
\item Because the vertices of $G$ are isolated points in $\mathcal{K}$, they
  all belong to the interior of $G$, i.e.\ $G$ is open in $\mathcal{K}$ --
  whereas leaves belong to $\partial G$.
\item Although $\mathcal{E} \subset \mathrm{cl}(G)$, not all limit points
  are ends: as Figure~\ref{figClosureGExample} shows, limit points can
  correspond to union of ends. \qedhere
\end{mathlist}
\end{remark}

\begin{figure}[h]
  \centering
  \captionsetup{width=.9\linewidth}
  \includegraphics[width=0.7\linewidth]{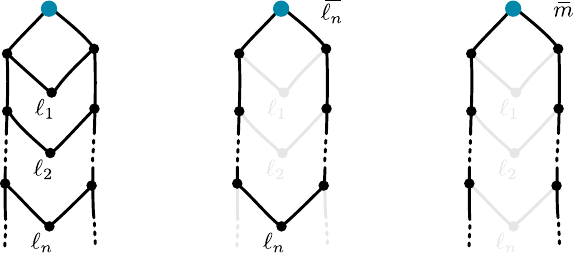}
  \caption{
  As previously, the root is highlighted in blue and the edges are pointing
  downwards.  In this example, the sequence $(\ell_n)$ has a limit
  $m \in \mathcal{K} \setminus \partial G$. Note that this limit corresponds to
  the union of two ends: in fact, although we do not detail this here, 
  it is not too hard to show that the limit in $\mathcal{K}$  of a
  sequence of points of $\partial G$ is always either a leaf or a union of
  ends.}
  \label{figClosureGExample}
\end{figure}

We close this appendix by briefly discussing the connection with a standard
notion of boundary for transient random walks: the Martin boundary.
Let us start by recalling its definition in our setting; we refer the reader
to \cite[Chapter~10]{kemeny1976denumerable} for a more general definition.

Let $\Prob[u]{\,\cdot\,}$ denote probabilities conditional on the
directed random walk $X$ being started from $u$, instead of from the root.
The \emph{Martin kernel of $X$} is defined as
\[
  M(u, v) \;=\;
  \frac{\Prob[u]{v \in X}}{\Prob{v \in X}} \,.
\]
Note that the functions $M(u,\, \cdot\,)$ are bounded, since
\[
  \frac{\Prob[u]{v \in X}}{\Prob{v \in X}}
  \;=\; \frac{\Prob{u \in X \given v \in X}\, \Indic{u \ancestor v}}{\Prob{u \in X}}
  \;\leq\; \frac{1}{\Prob{u \in X}}\,.
\]
As a result, there is a canonical compactification of $G$ on which all
functions $M(u, \,\cdot\,)$ extend continuously (see e.g.\
\cite[Theorem~7.3]{woess2009denumerable}).
This compactification, which we denote by $\widehat{G}$, is known as the Martin
compactification of $G$. The \emph{Martin boundary of $G$} is then defined as
$\widehat{G} \setminus G$.

Equivalently, a sequence $(v_n)$ of vertices of $G$ converges 
in $\widehat{G}$ if and only if $M(u, v_n)$ converges for all $u\in G$.
The Martin boundary consists of all limit points of sequences $(v_n)$ that
escape to infinity (in the sense that $h(v_n) \to \infty$, where $h(v)$ denotes
the height of $v$).

Although the Martin boundary and our boundary $\partial G$ are both designed to
capture the asymptotic behavior of $X$ (and both make it possible to define an
almost-sure limit for $X_t$), they are distinct notions and there does not seem
to be a simple relationship between them; in particular, one is not finer
than the other.

\begin{figure}[h]
  \centering
  \captionsetup{width=.9\linewidth}
  \includegraphics[scale=0.6]{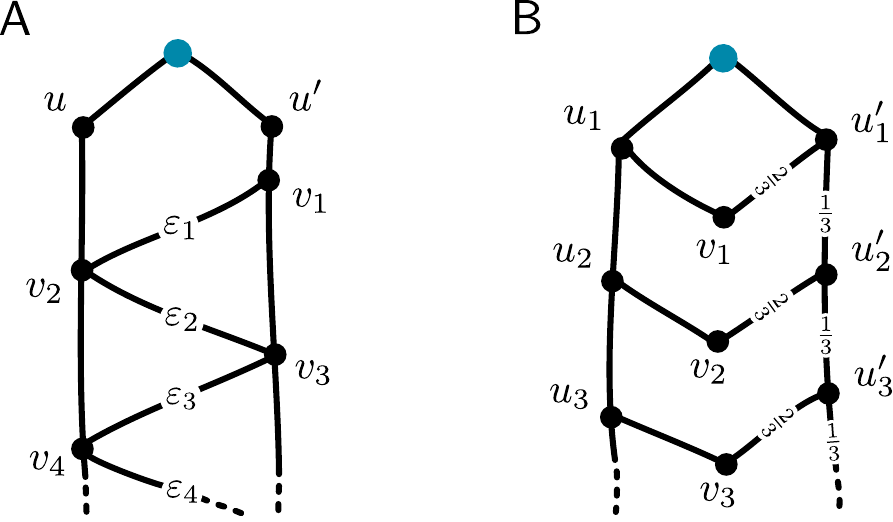}
  \caption{Two phylogenetic networks for which the Martin boundary differs
  from $\mathcal{K}$.  Again, the root is highlighted in blue and the edges are
  pointing downwards.  Provided that the weights of the edges going out of a
  vertex are rational probabilities, by replacing that vertex by a network as
  in Section~\ref{secBlowUps}, we can assume that the probability that the
  directed random walk follows an edge is equal to its weight.}
  \label{figMartinBoundary}
\end{figure}

For instance, the phylogenetic network~$G$ given in
Figure~{\ref{figMartinBoundary}.A} provides an example of a sequence of
vertices that converges in $\mathcal{K}$ to an end $x \in \partial G$, but that
does not converge in $\widehat{G}$: indeed, assume that the weights
$(\epsilon_n)$ given in Figure~\ref{figMartinBoundary}.A satisfy
\[
  \delta \;\defas\; \sum_{n\geq 1} \epsilon_n \;\leq\; \frac{1}{4}\,.
\]
Let $A$ be the event that the directed random walk $X$ started from the root
of $G$ goes through one of the central edges $(v_1\to v_2), (v_2\to v_3), \ldots$,
that is
\[
  A \;=\; \bigcup_{k \geq 1} \Set*[\big]{v_k \in X \text{ and } v_{k+1} \in X} \,.
\]
Using union bounds, we see that
\[
  \max\Set*[\big]{\Prob{A},\; \Prob{A \given u\in X},\;
    \Prob{A \given u'\in X}} \;\leq\; \delta\,.
\]
If $X$ goes to $u$ and then never crosses to the right through one of the
central edges, then it goes through $v_{2n}$. As a result, letting
$\Complement*{A}$ denote the complement of $A$,
\[
  \Prob*{\big}{u\in X, v_{2n}\in X} \;\geq\;
  \Prob*{\big}{\{u\in X\}\cap \Complement*{A}} \;\geq\;
  \frac{1}{2}(1-\delta)\,.
\]
Conversely, if $X$ goes to $u$ and then visits $v_{2n+1}$, then it must have
crossed to the right using one of the central edges -- so that
\[
  \Prob*{\big}{u\in X, v_{2n+1}\in X} \;\leq\;
  \Prob*{\big}{\{u\in X\}\cap A} \;\leq\; \frac{1}{2}\,\delta\,.
\]
Similarly, we get $\Prob{u'\in X, v_{2n}\in X} \leq \frac{1}{2}\delta$ and
$\Prob{u'\in X, v_{2n+1}\in X} \geq \frac{1}{2}(1-\delta)$. As a result,
\[
  \begin{dcases}
  \;\Prob*{\big}{u\in X\given v_{2n}\in X} \;\geq\; \frac{1}{2}(1-\delta) \;\geq \; \frac{3}{8} \\
  \;\Prob*{\big}{u\in X\given v_{2n+1}\in X} \;\leq\; \frac{\delta}{1-\delta} \;\leq \; \frac{1}{3}\, ,
  \end{dcases}
\]
from which it follows that $(v_n)$ does not converge in $\widehat{G}$ -- whereas it does in $\mathcal{K}$,
since $\partial G$ consists of a single end.

Conversely, the phylogenetic network $G'$ given in
Figure~\ref{figMartinBoundary}.B provides an example of a sequence
that converges in the Martin boundary $\widehat{G}'$, but not in $\mathcal{K}'$:
indeed, it is readily computed that
\[
  \Prob*{\big}{u_1\in X\given v_{n}\in X} \;=\;
  \frac{(\frac{1}{2})^{n+1}}{(\frac{1}{2})^{n+1}+(\frac{1}{3})^{n}}
  \;\tendsto{n\to\infty}\; 1,
\]
from which it follows that the sequence $(u_1,v_1,u_2,v_2,\dots)$ converges in
$\widehat{G}'$. Yet this sequence does not converge in $\mathcal{K}$,
because the subsequences $(u_n)$ and $(v_n)$ have different limits.

\subsection{Properties of the entropy} \label{appPropEntropy}

To make this article self-contained, in this appendix we  recall the
definition of the entropy of a probability measure on a general (i.e.\ not
necessarily countable) space, as well as some of its main properties. The
results are presented mostly without proofs, and we refer the reader to
\cite[Chapter~2]{martin1984mathematical} for a detailed treatment.

Let us start by recalling the definition of the entropy of a countable
partition.

\begin{definition} \label{defCountableEntropy}
Let $(\Omega, \mathscr{A}, \mu)$ be a probability space, and let $\pi$ be
a countable measurable partition of $\Omega$. Then, the \emph{entropy of $\pi$
with respect to $\mu$ is}
\[
  H_\mu(\pi) \;=\;
  - \sum_{A \in \pi} \mu(A) \log_2 \mu(A) \,.\qedhere
\]
\end{definition}

A key property of the entropy, which we refer to as its \emph{monotonicity},
is given by the next proposition.  Recall that a partition $\pi$ is said to be
\emph{finer} than a partition~$\pi'$ if for all $B \in \pi$ there exists $B'\in
\pi'$ such that $B \subset B'$.  We write $\pi \finer \pi'$ to indicate that
the partition $\pi$ is finer than $\pi'$.

\begin{proposition}[Monotonicity of the entropy] \label{propFundEntropy}
Let $\pi$ and $\pi'$ be two [countable] measurable partitions.
If $\,\pi' \finer \pi$, then $H_\mu(\pi') \geq H_\mu(\pi)$.
\end{proposition}

In particular, Proposition~\ref{propFundEntropy} justifies -- and is trivially
preserved by -- the following extension of Definition~\ref{defCountableEntropy}
to uncountable partitions.

\begin{definition} \label{defGeneralEntropy}
Let $(\Omega, \mathscr{A}, \mu)$ be a probability space, and let $\pi$ be
a measurable partition of $\Omega$. The \emph{entropy of $\pi$ with respect
to $\mu$} is
\[
  H_\mu(\pi) \;=\;
  \sup \Set*[\big]{H_\mu(\pi') \suchthat \pi'
  \text{ countable measurable partition of $\Omega$ s.t.\ } \pi \finer \pi'}\,. 
\]
The \emph{entropy of the probability distribution $\mu$} is then defined as
\[
  H(\mu) \;=\;
  \sup \Set*[\big]{H_\mu(\pi) \suchthat \pi \text{ measurable partition of $\Omega$}}
  \,. \qedhere
\]
\end{definition}

The monotonicity of the entropy entails that, letting
$\sigma$ denote the partition into singletons, $H(\mu) = H_\mu(\sigma)$.
In particular, if $\Omega$ is countable then we recover the familiar 
definition of Shannon entropy:
\[
  H(\mu) \;=\; - \sum_{i \in \Omega}
  \mu(\{i\}) \log_2 \mu(\{i\}) \, .
\]
In fact, if $\mu$ has a non-atomic part, then its entropy is
infinite. This is sometimes used as an equivalent definition of the
entropy, see e.g.\ \cite[Definition~2.14]{martin1984mathematical},
but we will provide a short proof of this equivalence below. Before
stating it, let us point out an elementary but useful fact. Since the
proof only involves straightforward calculations, we omit it.

\begin{proposition} \label{propGraftingEntropy}
Let $(\Omega, \mathscr{A}, \mu)$ be a probability space, and let $\pi$ be
a measurable partition of~$\,\Omega$.
Assume that $\pi'$ is obtained from $\pi$ by fragmenting
one of its block~$B$ such that $\mu(B) > 0$, and let $\pi_B$ denote the
corresponding partition of $B$.  Then,
\[
  H_{\mu}(\pi') \;=\; H_\mu(\pi) \;+\; \mu(B) \, H_{\mu_B}(\pi_B) \,,
\]
where $\mu_B$ denotes the conditional probability distribution induced on $B$
by $\mu$.
\end{proposition}
  
\begin{proposition} \label{propShannonEntropyPolishSpace}
  Let $(\Omega, \mathscr{A}, \mu)$ be a probability space. If $\mu$
has a non-atomic part of positive mass, then $H(\mu) = +\infty$.
\end{proposition}

\begin{proof}
Let $(A_i)_{i\geq 1}$ denote the family of distinct atoms of $\mu$.
Let $A=\bigcup_i A_i$, and note that by assumption
$\mu(\Complement*{A}) > 0$. Define
\[
  \sigma_A \;=\; \Set{A_i \suchthat i \geq 1}\cup \{\Complement*{A}\} \,.
\]
By Proposition~\ref{propGraftingEntropy}, we have
\[
  H(\mu) \;=\;
  H_\mu(\sigma_A) \;+\; \mu(\Complement*{A})\, H(\mu_{\Complement*{A}}) \,.
\]
Thus, to finish the proof it suffices to show that
$H(\mu_{\Complement*{A}}) = +\infty$.

By Lyapunov's theorem on non-atomic measures (see
e.g.~\cite{Art90}), for all $n$ it is possible to find a measurable
partition $\pi_n = \{B_1, \dots, B_n\}$ of $\Complement*{A}$ such
that $\mu(B_i)=\mu(\Complement*{A})/n$ for all
$i\in \{1, \dots, n\}$. It follows that
\[
  H(\mu_{\Complement*{A}}) \;\geq\; H_{\mu_{\Complement*{A}}}(\pi_n) \;=\; \log_2(n) \,.
\]
Since $n$ can be arbitrarily large, this concludes the proof.
\end{proof}

Finally, we close this section by a useful lemma for refining sequences of
partitions.  For this, we first introduce some vocabulary and notations.  Given
a partition $\pi$, we denote by:
\begin{itemize}
  \item $\pi[x]$ the block of $\pi$ that contains $x$;
  \item $\sim_\pi$ the equivalence relation associated with $\pi$.
\end{itemize}

\begin{definition} \label{defCVPartitions}
We say that a sequence $(\pi_n)$ of partitions of a set $E$ \emph{converges
to the partition $\pi$} if, for all $x, y \in E$, there exists
$N$ such that, for all $n \geq N$,
\[
  x \,\sim_{\pi_n} y \;\iff\; x \,\sim_\pi y \, . \qedhere
\]
\end{definition}

\begin{remark} \label{remCVPartitions}
Equivalently, $(\pi_n)$ converges to $\pi$ if and only if
$\lim_n \pi_n[x] = \pi[x]$ for every $x \in E$, i.e.\ if and only if
$\Indic*{\pi_n[x]}$ converges pointwise to $\Indic*{\pi[x]}$ for all $x\in E$.
\end{remark}

\begin{lemma} \label{lemCVEntropy}
Let $(\pi_n)$ be a sequence of measurable partitions on a probability
space $(\Omega,\mathscr{A},\mu)$. If $(\pi_n)$ converges to a partition $\pi$,
such that $\pi\finer \pi_n$ for all $n\geq 0$, then $\pi$ is measurable and
\[
  H_\mu(\pi_n) \;\tendsto{n\to\infty}\; H_\mu(\pi).
\]
In particular, if $\pi$ is the partition into singletons, the limit is $H(\mu)$.
\end{lemma}

\begin{proof}
The fact that $\pi$ is measurable can be deduced from Remark~\ref{remCVPartitions}.
This remark is also the key to the rest of the proof.
Indeed, since $\pi_n[x] \to \pi[x]$ for all $x\in \Omega$,
\begin{equation} \label{eqMuPinToPi}
\mu(\pi_n[x]) \;\tendsto{n\to\infty}\; \mu(\pi[x]),\qquad \forall x\in \Omega.
\end{equation}
Now, for any measurable partition $\sigma$, define the following measurable
function:
\[
  f^\sigma\colon\, x\in \Omega\, \longmapsto
  \begin{cases}
    -\log_2\big(\mu(\sigma[x])\big) &\text{if } \mu(\sigma[x])>0, \\
    \infty & \text{otherwise}.
  \end{cases}
\]
By Equation~\eqref{eqMuPinToPi}, $f^{\pi_n}\to f^{\pi}$ pointwise.
Noting that, in view of Proposition~\ref{propShannonEntropyPolishSpace},
$H_\mu(\sigma)=\int f^\sigma \,d\mu$, we deduce from Fatou's lemma that
\[
  \liminf_{n\to\infty} H_\mu(\pi_n) \;\geq\; H_\mu(\pi).
\]
This concludes the proof, since, by Proposition~\ref{propFundEntropy},
$H_\mu(\pi_n) \leq H_\mu(\pi)$ for all $n$.
\end{proof}

\subsection{Properties of the \texorpdfstring{$B_2$}{B2} index} \label{appPropB2}

In this Appendix, we list and prove miscellaneous properties of the $B_2$ index.
Let us start with something elementary, but fundamental -- namely, the fact that
the grafting property holds for infinite phylogenetic networks.

\begin{proposition} \label{propGraftingInfinite}
Let $G_1$ and $G_2$ be two phylogenetic networks, and let $G$ be the
phylogenetic network obtained by identifying a leaf $\ell \in G_1$ with the
root of $G_2$. Then,
\[
  B_2(G) \;=\; B_2(G_1) \;+\; p_\ell\, B_2(G_2) \,, 
\]
where $p_\ell$ denotes the probability that the directed random walk on $G_1$
ends in~$\ell$.
\end{proposition}

Since the proof of Proposition~\ref{propGraftingInfinite} is a straightforward
application of the analogous Proposition~\ref{propGraftingEntropy} for the
entropy, we do not detail it.

%\begin{proof}
%Note that $\partial G$ can be partitioned into
%$\Set{\partial G_1 \setminus \{\ell\},\, \partial G_2}$, and let $\pi$ be the
%partition of $\partial G$ defined by
%\[
%  \pi \;=\; \Set{\{x\} \suchthat x \in \partial G_2 \setminus \{\ell\}}
%  \cup \Set{\partial G_2}\,.
%\]
%Thus, the partition of $\partial G$ into singletons can be obtained from $\pi$
%by fragmenting $\partial G_2$ into singletons. Moreover, because $\ell$ is a
%cut-vertex in $G$, the conditional probability distribution induced on
%$\partial G_2$ by the directed random walk started from the root of $G$ is the
%same as the probability distribution induced on $\partial G_2$ by the directed
%random walk started from the root of $G_2$. As a result, the proposition
%follows from Proposition~\ref{propGraftingEntropy} for the entropy.
%\end{proof}

Next, let us prove the measurability of $B_2$.  This technicality is needed to
ensure that the $B_2$ index of a random phylogenetic network is a well-defined
random variable.

\begin{proposition}
Let $\phylnet$ denote the space of phylogenetic networks, equipped with
the local topology and its Borel $\sigma$-field.
The function $B_2\colon \phylnet \to \R\cup\{+\infty\}$ is measurable.
\end{proposition}

\begin{proof}
We will show that $B_2$ can be expressed as a pointwise limit of continuous
functions from $\phylnet$ to $\R$.

Let $G$ be a fixed phylogenetic network, and let $X=(X_n)_{n\geq 0}$ be the
directed random walk on $G$, started from the root.  Note that $X$ is the
only source of randomness, and that, for any $n \geq 0$, the distribution of
$X_n$ depends only on $[G]_n$. Therefore, the functions
\[
  B_2^{k,n}\colon G \;\longmapsto\; -\!\!\!\sum_{A\subset [G]_k}\!\! p_A^n \log_2(p_A^n),
  \quad \text{where } p_A^n = \Prob{[\bar{X}_n]_k = A},
\]
are continuous for the local topology (recall that $\bar{X}_n$ denotes the set
of ancestors of $X_n$ in $G$, and that we use the short notation
$[\bar{X}_n]_k = \bar{X}_n \cap [G]_k$).  Furthermore, we gave in
Definition~\ref{defK} a compactification $\mathcal{K}$
of $G$ in which $X_n$ converges almost surely to $X_\infty$ and in which, 
by definition, the functions $x \mapsto [\bar{x}]_k$ are continuous,
for all~$k$. Therefore, for any $k\geq 1$ and $A \subset [G]_k$, by dominated
convergence,
\[
  p_A^n \;=\; \Prob{[\bar{X}_n]_k = A} \;\tendsto{n\to\infty}\;
  \Prob{[\bar{X}_\infty]_k = A} \;\asdef\; p_A\, .
\]
This implies that
\[
  B_2^{k,n}(G) \;\tendsto{n\to\infty}\;
  B_2^k(G) \;\defas\; -\!\!\!\sum_{A\subset [G]_k}\!\! p_A \log_2(p_A)
  \;=\; H_{\mu}(\pi_k),
\]
where $\mu$ is the distribution of $X_\infty$ on $\mathcal{K}$, and $\pi_k$ is
the partition defined by
\[
  x\sim_{\pi_k} y \;\iff \; [\bar{x}]_k = [\bar{y}]_k\,.
\]
Because the map $x\mapsto \bar{x}$ is injective (see
Proposition~\ref{propK}), as $k\to\infty$ the partition~$\pi_k$ tends to the
partition into singletons.
Therefore, by Lemma~\ref{lemCVEntropy} and
definition of $B_2$,
\[
  B_2^k(G) = H_\mu(\pi_k) \;\tendsto{k\to\infty}\; H(\mu) = B_2(G),
\]
concluding the proof.
\end{proof}

We now provide tools to study the continuity of $B_2$.
We focus on easy-to-check sufficient conditions that are likely to
hold for biologically relevant models of phylogenetic networks.

\begin{proposition} \label{propInequalityLimInfB2}
Let $(G_n)$ be a sequence of phylogenetic networks converging locally to $G$.
Assume that the directed random walk on $G$ ends in a leaf almost surely.
Then,
\[
  \liminf_{n\to\infty} B_2(G_n) \;\geq\; B_2(G).
\]
\end{proposition}

\begin{proof}
Let $\mathcal{L}$ be the leaf set of $G$, and let
$\mathcal{L}_k = \Set{\ell \in \mathcal{L} \suchthat h(\ell) \leq k}$
be the set of leaves at height at most~$k$.  As always, write $X$ for the
directed random walk on~$G$, and $X_\infty$ for its limit on the boundary of~$G$. 
Finally, let $p_\ell = \Prob{X_\infty=\ell}$.

Since $\Prob{X_\infty \in \mathcal{L}} = 1$, we have
$B_2(G) = -\!\sum_{\ell \in \mathcal{L}} p_\ell \log_2 p_\ell$; and since
$\mathcal{L} = \limup_k \mathcal{L}_k$, it follows that
\[
  -\sum_{\ell \in \mathcal{L}_k} \! p_{\ell}\log_2 p_{\ell}
  \;\tendsto{k\to\infty}\; B_2(G)\,.
\]
Now, fix some $b$ such that $0< b <B_2(G)$, and let $k$ be such that the sum in the
display above be greater than $b$.
Fix $n_0$ such that $[G_n]_k=[G]_k$ for all $n\geq n_0$.
Then, for all $n\geq n_0$,
\[
  B_2(G_n) \;\geq\; -\sum_{\ell \in \mathcal{L}_k} p_{\ell}\log_2 p_{\ell} \;>\; b.
\]
This shows that $\liminf_n B_2(G_n) \geq b$. Letting $b\to B_2(G)$ concludes the proof.
\end{proof}

\begin{proposition} \label{propInequalityLimSupB2}
Let $(G_n)$ be a sequence of \emph{finite}\! phylogenetic networks that
converge locally to~$G$.  Assume that, for a sequence $(k_n)$ satisfying
$[G_n]_{k_n}=[G]_{k_n}$, we have
\[
  \Prob*{\big}{h(X_\infty) > k_n}\,\log \Abs{G_n} \;\tendsto{n\to\infty}\; 0\,,
\]
where $h(X_\infty)$ is the height of the limit of the directed random walk on $G$,
and $\Abs{G_n}$ denotes the number of leaves of $G_n$. Then, 
\[
  \lim_{n\to\infty} B_2(G_n) \;=\; B_2(G).
\]
\end{proposition}

\begin{proof}
Let us write $q_n = \Prob{h(X_\infty) > k_n}$.
Since $\log \Abs{G_n}$ cannot tend to $0$ (except in the trivial case where $G$
is a finite DAG with a single leaf), we must have $q_n\to 0$.
This means that $\Prob{X_\infty \in \mathcal{L}} = 1$, where $\mathcal{L}$
denotes the leaf set of $G$. As a result, by
Proposition~\ref{propInequalityLimInfB2} it suffices to show
\[
  \limsup_{n\to\infty} B_2(G_n) \;\leq\; B_2(G) \;=\;
  -\sum_{\ell\in \mathcal{L}}p_{\ell} \log_2 p_\ell\,,
\]
where $p_\ell = \Prob{X_\infty = \ell}$.
Let us write $\mathcal{L}_{k_n}^{n}$ (resp.\ $\mathcal{L}_{>k_n}^n$) for
the set of leaves in $G_n$ at height at most $k_n$ (resp.~at least $k_n+1$).
For any leaf $\ell\in G_n$, write $\smash{p_\ell^{(n)}}$ for the
probability that the directed random walk on $G_n$ ends in $\ell$.  Notice
that, since $[G_n]_{k_n}=[G]_{k_n}$,
\[
  -\!\sum_{\ell \in \mathcal{L}_{k_n}^{n}}\!\! p^{(n)}_\ell \log_2 p^{(n)}_\ell \;=\;
  -\!\sum_{\ell \in \mathcal{L}_{k_n}}\!\! p_\ell \log_2 p_\ell \,,
\]
and that, since $(p^{(n)}_\ell\! / q_n : \ell \in \mathcal{L}_{>k_n}^{n})$ is
a probability distribution, Jensen's inequality gives
\[
  -\frac{1}{q_n}\;\sum_{\mathclap{\quad\ell \in \mathcal{L}_{>k_n}^{n}}}\;
      p^{(n)}_\ell \log_2p^{(n)}_\ell \;\leq\;
  \log_2\bigg(\;\sum_{\mathclap{\quad\ell \in \mathcal{L}_{>k_n}^{n}}}\;\;
      \frac{p^{(n)}_\ell}{q_n}\cdot\frac{1}{p^{(n)}_\ell} \bigg) \;\leq\;
  \log_2\Big(\frac{\Abs{G_n}}{q_n}\Big)\,.
\]
It follows that
\[
  B_2(G_n) \;\leq\; -\!\sum_{\ell \in \mathcal{L}_{k_n}}\!\! p_\ell \log_2 p_\ell
  \;+\;  q_n\log_2\Big(\frac{\Abs{G_n}}{q_n}\Big),
\]
where the first term on the right-hand side tends to $B_2(G)$, while the second
one tends to $0$ because $q_n\to 0$ and $q_n \log_2 |G_n|\to 0$ by hypothesis.
\end{proof}

Finally, we close this section by providing a useful tool to compute the $B_2$ index
of a fixed infinite network~$G$. This essentially consists in chopping the
network at its cut-vertices, and using the grafting property -- but there is a
small subtlety (namely, showing that the resulting nondecreasing sequence of
$B_2$ indices, which is trivially bounded above by the $B_2(G)$, actually
reaches it in the limit). In order to give a clean statement, let us introduce
some vocabulary.

\begin{definition} \label{defStubSeq}
A \emph{stub} of a phylogenetic network $G$ is any phylogenetic network
$S \subset G$ that can be obtained by removing all out-going edges from a set
of cut-vertices of $G$ and keeping the connected component containing the
root. We write $S \sqsubset G$ to indicate that $S$ is a stub of $G$.
A \emph{stub sequence} of $G$ is an increasing sequence of stubs
$S_n \sqsubset S_{n+1} \sqsubset G$ such that $\lim_n S_n = G$.
\end{definition}

\begin{proposition} \label{propStubSeq}
Let $(S_n)$ be a stub sequence of $G$. Then,
\[
  B_2(S_n) \;\tendsto{n\to\infty} \;B_2(G)\,.
\]
\end{proposition}

\begin{proof}
For any stub $S \stub G$, let $\pi_S$ be the partition of $\partial G$
defined by the following equivalence relation:
\[
  x \sim_{S} y \;\iff\;
  x = y \text{ or }
  \bar{x} \cap \bar{y} \text{ contains a leaf of } S.
\]
To see that this is indeed an equivalence relation, recall that the leaves
of $S$ are either leaves or cut-vertices in $G$, and note that for all
$x \in \partial G$, the set $\bar{x}$ contains at most one leaf of $S$.
The transitivity of $\sim_S$ then follows readily.

Note that the partitions $\pi_{S_n}$ are refining along the stub sequence
$(S_n)$, and that they converge simply (see Definition~\ref{defCVPartitions}) to
$\pi_G$, the partition into singletons.
Indeed, pick $x\neq y$ and assume that $x\sim_{S_n}y$ for
all~$n$. For each $n$, let then $v_n$ be the unique leaf of $S_n$ in
$\bar{x} \cap \bar{y}$. Then:
\begin{itemize}
  \item If $\Set{v_n \suchthat n \geq 1}$ is finite, then there exists $\ell$
    such that $v_n = \ell$ for all $n$ large enough. By the local
    convergence of $S_n$ to $G$, this vertex $\ell$ is a leaf of $G$, and
    therefore $x = y = \ell$, yielding a contradiction.
  \item If $\Set{v_n \suchthat n \geq 1}$ is infinite, then
    $v_1, v_2, \ldots$ are cut-vertices of $G$ and they lie on a ray
    $r \subset \bar{x} \cap \bar{y}$. This entails
    $x = y$: indeed, let $r_x \in x$ and $r_y \in y$. For all $n$, 
    since $v_n \in \bar{x}$ and $v_n$ is a cut-vertex of $G$, $v_n \in
    r_x$. As a result, $r_x$ intersects $r$ infinitely many
    times.  Similarly, $r_y$ intersects $r$ infinitely many times.
    Thus, we have $r_x \codir r_y$ and therefore $x = y$ -- again
    yielding a contradiction.
\end{itemize}

Finally, note that $B_2(S_n) = H_\mu(\pi_{S_n})$, where $\mu$ is the distribution of $X_\infty$, and that, by definition,
$B_2(G) = H_\mu(\pi_G)$. Thus, the proposition follows from
Lemma~\ref{lemCVEntropy}.
\end{proof}

\subsection[Total variation bound for the convergence of Galton–Watson trees]{%
Total variation bound for the convergence of \newline Galton–Watson trees} \label{appTotalVariationBound}

In this appendix, we provide a total variation bound to quantify the speed
of convergence of conditioned Galton--Watson trees to Kesten's tree. Let
us start by setting\,/\,recalling some notation.

In the remainder of this appendix, $\xi$ will denote a random variable taking values in
$\N = \Set{0, 1, 2, \ldots}$, with mean $1$ and such that $\Prob{\xi=0}>0$.
Likewise, $T \sim \mathrm{GW}(\xi)$ will be a Galton--Watson tree with
offspring distribution $\xi$, which we view as an ordered tree.
Recall that, for any tree $\mathbf{t}$, we denote the number of leaves of
$\mathbf{t}$ by~$\Abs{\mathbf{t}}$. As in Section~\ref{secBlowUps}, for any~$n$
such that $\Prob{\Abs{T}=n}>0$, we write $T_n$ for a random tree
distributed as $T$ conditioned to have $n$ leaves. Note that, whenever we state
something about~$T_n$, we implicitly assume that $n$ is taken so that
$\Prob{\Abs{T}=n}>0$. Finally, we denote by $T_*$ the Kesten tree
associated to $T$, i.e.\ the local limit of $T_n$ as $n\to\infty$ (see
Definition~\ref{defKesten}), and we write $\hat{\xi}$ for a random variable
such that $\Prob*{\normalsize}{\hat{\xi}=n}=n\,\Prob{\xi=n}$ for all~$n \geq 1$
(thus, $\hat{\xi}$ is distributed as the number of children of the vertices
that lie on the spine of~$T_*$).

Let us now introduce a truncation of $T_*$: for each $k\geq 0$, let $v_k$ be
the $k$-th vertex on the spine of $T_*$, with $v_0$ being the root, and let
$T_*^k$ be the leaf-pointed rooted tree obtained from $T_*$ by removing all
descendants of $v_k$ (other than $v_k$ itself) and letting $v_k$ be the
distinguished leaf.  The distribution of $T_*^k$ is easily computed:
let~$\mathbf{t}$ denote a finite ordered tree with a distinguished leaf $v_k$
at height $k$, and write $\mathbf{v}=(v_0,v_1,\dots,v_k)$ for the path from the
root of $\mathbf{t}$ to its distinguished leaf~$v_k$.  Then, by construction
of~$T_*$,
\begin{align*}
  \Prob{T_*^k=\mathbf{t}}
  \;&= \Bigg(\prod_{i=0}^{k-1}\Prob{\hat{\xi}=d_{v_i}}\cdot \frac{1}{d_{v_i}}\Bigg)
       \Bigg(\!\!\prod_{\;v\in \mathbf{t}\setminus \mathbf{v}}\!\! \Prob{\xi=d_{v}} \Bigg) \\[0.5ex]
  &=\;\; \prod_{\mathrlap{\hspace{-1em} v\in \mathbf{t}\setminus \{v_k\}}}\, \Prob{\xi=d_{v}} \\[1ex]
  &=\; \frac{\Prob{T=\mathbf{t}}}{\Prob{\xi=0}},
\end{align*}
where $d_v$ denotes the number of children of vertex $v$ and, by a
slight abuse of notation, we write $\{T=\mathbf{t}\}$ for the event that the
Galton--Watson tree $T$ is equal to the non-pointed version of the tree
$\mathbf{t}$. Moreover, note that ${\Prob{T=\mathbf{t}}}/ {\Prob{\xi=0}}$
is also the probability that $T$ ``starts'' with the leaf-pointed tree
$\mathbf{t}$, i.e.\ that $T$ can be obtained by grafting a tree $T'$ on the
distinguished leaf of $\mathbf{t}$.  Writing $\{\mathbf{t}\subset T\}$ for this
event, we therefore have
\[
  \Prob{T_*^k=\mathbf{t}} \;=\; \Prob{\mathbf{t}\subset T}\,.
\]
Also note that, conditional on $\{\mathbf{t}\subset T\}$, the subtree of $T$
descending from the distinguished leaf of $\mathbf{t}$ is a
$\mathrm{GW}(\xi)$ tree.

Next, let us build a pointed tree $T_n^k$ from $T_n$, as we did for
$T_*^k$. Conditional on~$T_n$, let $u_n$ be chosen uniformly at random among
the $n$ leaves of $T_n$.  If $u_n$ is at height at least $k$, then let
$v_{n,k}$ be the vertex at distance $k$ from the root on the path to~$u_n$, and
define $T_n^k$ to be the tree obtained from $T_n$ by removing all vertices
descending from $v_{n,k}$ (other than $v_{n, k}$ itself) and letting
$v_{n,k}$ be the distinguished vertex. On the event that $u_n$ is at height
less than $k$, the definition of $T_n^k$ is irrelevant; we define it to be the
tree reduced to a single node.  From this construction, note that for any
$\mathbf{t}$ finite ordered tree with a distinguished leaf at height $k$,
\begin{align}
  \Prob{T_n^k=\mathbf{t}} \;&=\; \frac{n-\Abs{\mathbf{t}}+1}{n}\cdot\frac{\Prob{\mathbf{t}\subset T, \, \Abs{T}=n}}{\Prob{\Abs{T}=n}} \nonumber\\
  &=\;\frac{n-\Abs{\mathbf{t}}+1}{n}\cdot \Prob{\mathbf{t}\subset T} \cdot \frac{\Prob{\Abs{T}=n-\Abs{\mathbf{t}}+1}}{\Prob{\Abs{T}=n}} \nonumber\\
  &=\; \Prob{T_*^k=\mathbf{t}} \cdot \frac{f(n-\Abs{\mathbf{t}}+1)}{f(n)}, \label{eqCalculTnk}
\end{align}
where $f(n)\defas n\,\Prob{\Abs{T}=n}$.

With these definitions, we are ready to state the main result of this appendix.

\begin{proposition}\label{propDTV}
If $\,\ExpecBrackets{\xi^3}<\infty$, then for any sequence of integers
$(k_n)_{n\geq 0}$ such that $k_n=o(\sqrt{n})$, we have
\[
  d_{\mathrm{TV}}\big(T_n^{k_n},\, T_*^{k_n}\big) \;=\; \Theta\Big(\frac{k_n}{\sqrt{n}}\Big)\,.
\]
\end{proposition}

\begin{remark}
The fact that the total variation distance between $T_n^{k_n}$ and $T_*^{k_n}$
tends to~$0$ whenever $k_n = o(\sqrt{n})$ already appears, in slightly
different forms, in the literature: a similar result, but where $T_n$ is
conditioned on its number of vertices, can be found in
Kersting~\cite[Theorem~5]{Ker11} -- without any moment assumption besides
$\ExpecBrackets{\xi} = 1$, and for any sequence $k_n=o(n/a_n)$ with $a_n$ being any
sequence such that $\sum_{i=1}^n\xi_i/a_n$ converges in distribution (where
$\xi_i$ are i.i.d.\ replicates of~$\xi$).
Stufler~\cite[Theorem~5.2]{stufler2019local} also states a closely related
result for fringe subtrees, under a finite variance assumption.
What is new here is the explicit speed of convergence $k_n/\sqrt{n}$, which we
need to show the convergence of all moments of $B_2(G_n)$ to those of
$B_2(G_*)$ in the proof of Theorem~\ref{thmContinuityB2BlowUps}.

The third moment assumption is needed in our proof to get a Berry--Esseen
type estimate, but this may be a superfluous assumption.
It would be interesting to see whether similar bounds can be obtained in the
critical case without any moment assumption (similar to Kersting), but this 
goes beyond the scope of this paper.

Finally, note that the proof that we develop would be essentially the same if
$T_n$ were conditioned on the number of vertices, or if we looked at fringe
subtrees.
\end{remark}

The proof of Proposition~\ref{propDTV} relies on the following lemma.

\begin{lemma} \label{lemEstimatesForLeaves} 
With the same notation as above,
\begin{mathlist}
  \item If $\,\ExpecBrackets{\xi^3}<\infty$, then there exists a constant $c_1 > 0$
    such that, along any sequence of $n$ such that $\Prob{\Abs{T}=n}>0$,
    \[
      f(n) \;=\; n\,\Prob{\Abs{T}=n} \;=\; \frac{c_1}{\sqrt{n}} + O(n^{-1}).
    \]
  \item If $\ExpecBrackets{\xi^2}<\infty$, then for any sequence of integers
    $(k_n)_{n\in \mathbb{N}}$ such that $k_n=o(\sqrt{n})$, we have
    \[
      \Prob{\Abs[\normalsize]{T_*^{k_n}} \geq n} \;=\;
      \Theta\Big(\frac{k_n}{\sqrt{n}}\Big).
    \]
\end{mathlist}
\end{lemma}

\begin{proof}
(i) By Minami's correspondence \cite[Theorem~2]{Min05}, we have
\[
  \Prob{\Abs{T}=n} \;=\; \Prob{v(T')=n},
\]
where $v(T')$ denotes the number of vertices of $T'$, and $T'$ is
a $\mathrm{GW}(\xi')$ tree whose offspring distribution is
\[
  \xi' \;=\; \sum_{i=1}^{Z}Y_i \,,
\]
with $\Prob{Z=n} = \Prob{\xi>0}^n\, \Prob{\xi=0}$ for all $n\geq 0$, and where
the variables $Y_i$ are i.i.d., independent of $Z$, and distributed as
$(\xi-1 \;|\; \xi>0)$.  It is easy to see that $T'$ is critical and that
$\xi'$ has a finite third moment.

Next, we use the well-known fact (see e.g.~\cite{Dwa69}) that
\[
  n\, \Prob{v(T')=n} \;=\;\Prob{S_n=n-1}, \quad \text{with } S_n = \sum_{i=1}^{n} \xi'_i,
\]
where the $\xi'_i$ are i.i.d.\ copies of $\xi'$. Since $\xi'$ has a finite
third moment, the asymptotic expansion given in~(i) follows readily
from a Berry--Esseen type estimate for the local central limit
theorem~\cite[Chapter~VII, Theorem~6]{Pet75}.

(ii) Let us fix a sequence $k_n = o(\sqrt{n})$.
By construction of Kesten's tree, we have
\begin{equation} \label{eqLTKasSum}
  \Abs{T_*^{k_n}}-1 \;\,\overset{d}{=}\,\; \sum_{i=1}^{k_n} X_i,
\end{equation}
where the $X_i$ are independent and distributed as $\sum_{i=1}^{Y}\Abs{T_i}$,
where $Y$ is distributed as $\hat{\xi}-1$ and the $T_i$ are i.i.d.\
$\mathrm{GW}(\xi)$ trees.

Similarly to (i), using a standard local central limit theorem -- this time,
requiring only a finite second moment for $\xi$ -- we get
$\Prob{\Abs{T}\geq n} \sim C_1n^{-1/2}$ for some constant $C_1>0$ (for a
complete proof in a more general case, see e.g.\ \cite[Theorem~3.1]{Kor12}).
Because $\hat{\xi}-1$ has a finite expectation, we also have
$\Prob{X_1\geq n} \sim C_2n^{-1/2}$ for the variables $X_i$
in~\eqref{eqLTKasSum}, for some constant $C_2>0$.  Therefore, $X_1$ is in the
domain of attraction of a $(1/2)$-stable distribution and we can apply Heyde's
large deviation theorem~\cite{Hey67}: for any nondecreasing sequence $(x_m)$
such that $\sum_{i=1}^m X_i/x_m \to 0$ in probability, we have
\[
  \Prob{\sum_{i=1}^m X_i \geq x_n} \;=\;
  \Theta\big(m\,\Prob{X_1 \geq x_n}\big) \;=\;
  \Theta\Big(\frac{m }{\sqrt{x_n}}\Big).
\]
Since $k_n=o(\sqrt{n})$ and since $X_1$ is in the domain of a
$(1/2)$-stable distribution, we have $\sum_{i=1}^{k_n} X_i/n \to 0$ in
probability; and so the equation above translates to
\[
  \Prob{|T_*^{k_n}| \geq n} \;=\;
  \Prob{\sum_{i=1}^{k_n} X_i \geq n} \;=\;
  \Theta\Big(\frac{k_n}{\sqrt{n}}\Big),
\]
concluding the proof.
\end{proof}

\begin{proof}[Proof of Proposition~\ref{propDTV}]
Recall from~\eqref{eqCalculTnk} that for any fixed tree $\mathbf{t}$ with a
distinguished leaf at height $k$, we have
\[
  \Prob{T_n^k=\mathbf{t}} \;=\;
  \Prob{T_*^k=\mathbf{t}}\,\frac{f(n-\Abs{\mathbf{t}}+1)}{f(n)},
\]
with $f(n)=n\Prob{\Abs{T}=n}$. This implies that, letting $(\cdot)_+$
denote the positive part,
\begin{align}
  d_{\mathrm{TV}}\big(T_n^k,\, T_*^k\big) 
  &\;=\; \sum_{\mathbf{t}}\big(\Prob*{\normalsize}{T_*^k=\mathbf{t}} -
                               \Prob*{\normalsize}{T_n^k=\mathbf{t}}\big)_+ \nonumber \\[1ex]
  &\leq\; \Prob{|T_*^k| > \tfrac{n}{2}} \;+ \hspace{-0.75em}
          \sum_{\quad\mathbf{t}:\Abs{\mathbf{t}}\leq n/2} \hspace{-1.4em}
             \big(\Prob*{\normalsize}{T_*^k=\mathbf{t}} - \Prob*{\normalsize}{T_n^k=\mathbf{t}}\big)_+ \nonumber \\
  &=\; \Prob{|T_*^k| > \tfrac{n}{2}} \;+\;
       \ExpecBrackets{\Big(1-\frac{f(n-|T_*^k|+1)}{f(n)}\Big)_{\!+}\!\! \cdot\, \Indic{|T_*^k|\leq \tfrac{n}{2}}}. \label{eqDTV1}
\end{align}
Using Lemma~\ref{lemEstimatesForLeaves}~(i), we see that as $n\to\infty$ and
uniformly for all $\ell \leq n/2$ such that $f(n-\ell+1)>0$,
\[
  \frac{f(n-\ell+1)}{f(n)} \;=\; 
  \frac{\frac{c_1}{\sqrt{n - \ell + 1}} + O(n^{-1})}{\frac{c_1}{\sqrt{n}} + O(n^{-1})} \;=\;
  \sqrt{\frac{n}{n - \ell + 1}} \;+\; O(n^{-1/2}) \,.
\]
Because $(\frac{n}{n-\ell+1})^{1/2}\geq 1$, this gives
\[
  \Big(1-\frac{f(n-\ell+1)}{f(n)}\Big)_{\!+} \;=\,\; O(n^{-1/2}) \,, 
\]
and by plugging this into~\eqref{eqDTV1} we obtain
\[
  d_{\mathrm{TV}}\big(T_n^k,\, T_*^k\big) \;\leq\;
  \Prob{|T_*^k| > \tfrac{n}{2}} \;+\; O(n^{-1/2}).
\]
Finally, since $T_n$ has exactly $n$ leaves, the following lower bound is immediate:
\[
  d_{\mathrm{TV}}\big(T_n^k,\, T_*^k\big) \;\geq\; \Prob{|T_*^k| > n}.
\]
Considering a sequence $k_n=o(\sqrt{n})$ and using Lemma~\ref{lemEstimatesForLeaves}~(ii) readily concludes the proof of the proposition.
\end{proof}

\subsection{Blowups of Galton--Watson trees: proofs} \label{appProofsBlowUps}

In this appendix, we prove the results stated in Section~\ref{secBlowUps} of
the main text. Most of the appendix is devoted to the proof of
Theorem~\ref{thmContinuityB2BlowUps}, but more elementary and general results
about the $B_2$ index of blowups of (not necessarily random) trees are also
given.

Throughout this appendix, we use the notation and assumptions of
Appendix~\ref{appTotalVariationBound}: $T\sim \mathrm{GW}(\xi)$ is
a Galton--Watson tree with offspring distribution $\xi$, and we assume
that $\xi$ is critical, with $\Prob{\xi=0}>0$ and $\ExpecBrackets{\xi^3}<\infty$.
We denote by $T_n$ (resp.\ $T_*$) the corresponding tree conditioned on
having $n$ leaves (resp.\ Kesten tree).

We also use the notation of Section~\ref{secBlowUps}: $G_n$ and $G_*$ will
denote blowups of $T_n$ and $T_*$ with respect to a fixed sequence of
distributions $\nu=(\nu_k)_{k\geq 1}$. Finally, for any $G$ blowup of $T$ and for
each vertex $u\in T$, we let $\Head{u} \sim \nu_{d^+(u)}$ denote the random
network associated to $u$ in the blowup construction of $G$.

In the statement of the next lemma, let us write $(u_i)_{i\geq 0}$ for the
vertices along the spine of~$T_*$, and for each $i\geq 0$ let
$v_i$ be the root of $\Gamma_{u_i}$.

\pagebreak

\begin{lemma} \label{lemPkSmall}
Let $p_k=p_{v_k}$ denote the probability (conditional on $G_*$) for the directed
random walk on $G_*$, started from the root, to reach $v_k$. Then,
\begin{mathlist}
\item $\ExpecBrackets{p_k} = (1-\eta_0)^k$.
\item $p_k = O((1-\eta_0)^k)$ almost surely.
\end{mathlist}
where $\eta_0 = \Prob{\xi =0} >0$.
\end{lemma}

\begin{proof}
For $k\geq 0$, let $G_*^k$ denote the almost surely finite network
obtained from $G^*$ by removing all (strict) descendants of $v_k$ -- with the
notation of Appendix~\ref{appTotalVariationBound}, this is a blowup
of~$T_*^k$ with respect to~$\nu$.  Let us now fix $k\geq 0$ and define $D$ to
be the out-degree of~$u_k$, i.e.\ the number of leaves of $\Gamma_{u_k}$. For
each $i = 1, \ldots, D$, let then $q_i$ denote the probability, conditional on
$\Gamma_{u_k}$, for a random walk started from $v_k$ to pass through the
$i$-th leaf of $\Gamma_{u_k}$. Because $\Gamma_{u_k}$ is leaf-exchangeable,
we have
\[
  \ExpecBrackets{q_1 \given G_*^k, D} \;=\; \frac{1}{D}\,, 
\]
and, because $u_{k+1}$ is chosen uniformly among the children of $u_k$,
\[
  \ExpecBrackets{\frac{p_{k+1}}{p_k}\given G_*^k, D} \;=\;
  \ExpecBrackets{q_1 \given G_*^k, D} \;=\; \frac{1}{D}.
\]
Integrating with respect to $D$, we get
\[
  \frac{1}{p_k}\,\ExpecBrackets{p_{k+1}\given G_*^k} \;=\;
  \sum_{n\geq 1} n\,\Prob{\xi=n}\cdot \frac{1}{n} \;=\;
  \Prob{\xi \geq 1} \;=\; 1-\eta_0.
\]
Because the out-degrees of the vertices are i.i.d.\ on the spine of $T_*$,
it follows that $(1-\eta_0)^{-k}p_k$ is a positive martingale. As a result,
it converges almost surely to a nonnegative random variable. This concludes
the proof.
\end{proof}

We are now ready to prove Theorem~\ref{thmContinuityB2BlowUps} from the main
text, whose statement we recall here for convenience.

\begin{reptheorem}{thmContinuityB2BlowUps}[repeated from Section~\ref{secBlowUps}]
With the notation above, assuming that $\xi$ satisfies
$\ExpecBrackets{\xi} = 1$, $\Prob{\xi = 0} > 0$ and $\ExpecBrackets{|\xi|^3} < +\infty$,
we have:
\begin{mathlist}
  \item $B_2(G_n) \to B_2(G_*)$ in distribution.
  \item For all $p\geq 1$,
  $\ExpecBrackets{B_2(G_n)^p} \to \ExpecBrackets{B_2(G_*)^p}$,
  and all these moments are finite.
\end{mathlist}
\end{reptheorem}

\begin{proof}
The proof of the convergence in distribution consists in building a coupling of
$(G_n)_{n\geq 1}$ and $G_*$ such that any subsequence of $(B_2(G_n))_{n\geq 0}$
has a subsequence that converges to $B_2(G_*)$ almost surely.
To prove the convergence of moments, we will then
show that, for all $m \geq 1$, the sequence
$(\ExpecBrackets{B_2(G_n)^m})_{n \geq 1}$ is bounded.

Let us describe our coupling: first, pick any sequence of integers
$(k_n)_{n\geq 1}$ such that
\[
  \log\log n \;\ll\; k_n \;\ll\; n^{1/2 - \epsilon}, 
\]
for some $\epsilon > 0$.  Then, build $(G_*, T_*)$ as in
Section~\ref{secBlowUps}: let $\mathcal{U}$ denote the Ulam--Harris tree, and
let $\Gamma = (\Head{v}^{k} : v \in \mathcal{U}, k \geq 1)$ be a family of
independent phylogenetic networks such that $\Head{v}^{k} \sim \nu_{k}$.
Finally, sample $T_*$ independently of $\Gamma$, and let $G_*$ be the blowup
obtained from $(T_*, \Gamma)$.

Next, conditional on $T_*$, build a sequence $(G_n, T_n)_{n\geq 1}$ as follows:
for each $n \geq 1$, sample a conditioned $\mathrm{GW}(\xi)$ tree $T_n$ with
$n$ leaves, independently of $\Gamma$ and in such a way that
\[
  \Prob{T_n^{k_n} \neq T_*^{k_n}} \;=\;
  d_{\mathrm{TV}}\big(T_n^{k_n}, T_*^{k_n}\big)\,, 
\]
where the truncated trees $T_n^{k}$ and $T_*^k$ are defined as in
Appendix~\ref{appTotalVariationBound}.
Finally, let $G_n$ be the blowup obtained from $T_n$ and $\Gamma$.

Let us now show that any subsequence of $(B_2(G_n))_{n \geq 1}$
has a subsequence that converges to $B_2(G_*)$ almost surely -- and, therefore,
in distribution. By a standard result (see e.g.~\cite[Theorem~2.6]{billingsley1999convergence}), this will prove that
  $B_2(G_n) \to B_2(G_*)$ in distribution.
% Attention, c'est pas forcément vrai que "$B_2(G_n) \to B_2(G_*)$ almost surely" au global.
First, note that since
$k_n = o(\sqrt{n})$, by Proposition~\ref{propDTV} we have
\[
  d_{\mathrm{TV}}\big(T_n^{k_n}, T_*^{k_n}\big) \;=\;
  O\Big(\frac{k_n}{\sqrt{n}}\Big) \;\tendsto{n\to\infty}\; 0\,.
\]
As a result, along any increasing sequence of integers there exists
a subsequence, whose range we denote by $A$, such that
\[
  \sum_{n\in A} \Prob{T_n^{k_n} \neq T_*^{k_n}} \;<\; \infty\,.
\]
Thus, by the Borel--Cantelli lemma, there exists a random $n_0$ such that,
for all $n\in A\cap \COInterval{n_0,\infty}$, we have $T_n^{k_n}=T_*^{k_n}$
-- and, therefore, $G_n^{k_n}=G_*^{k_n}$.

As previously, for each $k$ let $v_k$ denote the distinguished leaf of
$G_n^{k}$ and let $p_k = p_{v_k}$ denote the probability that the directed
random walk reaches it. Since on the event $\Set*{G_n^{k_n}=G_*^{k_n}}$, the
network $G_n$ can be obtained from $G_*^k$ by grafting a network with at most
$n$ leaves on~$v_k$, by Proposition~\ref{propGraftingFinite} we have
\[
  B_2(G_*^{k_n}) \;\leq\; B_2(G_n) \;\leq\;
  B_2(G_*^{k_n}) \;+\; p_{k_n} \log_2 n.
\]
Because $k_n \gg \log\log n$, by Lemma~\ref{lemPkSmall}
the term $p_{k_n} \log_2 n$ vanishes almost surely.
Finally, by Proposition~\ref{propStubSeq}, $B_2(G_*^{k_n}) \to B_2(G_*)$
almost surely. As a result, $B_2(G_n) \to B_2(G_*)$ along the
subsequence indexed by $A$ -- concluding the proof of point~(i).

It remains to show the convergence of all moments of $B_2(G_n)$.
Because we have already proved the convergence in distribution, it is sufficient
to show that for all $m\geq 1$, $\ExpecBrackets{B_2(G_n)^m}$ is bounded: indeed, this
implies that $(B_2(G_n)^m)_{n\geq 1}$ is uniformly integrable
-- which, together with the convergence in distribution, implies the
convergence of moments (see e.g.~\cite[Lemma~5.11]{kallenberg2021foundations}).

Now, note that
\[
  B_2(G_n)^m \;\leq\; \Indic{G_n^{k_n}\neq G_*^{k_n}}(\log_2n)^m \;+\;
  \Indic{G_n^{k_n}= G_*^{k_n}} \big(B_2(G_*^{k_n}) + p_{k_n}\log_2 n\big)^m .
\]
Using $(a+b)^m \leq 2^{m-1}(a^m + b^m)$ for $a, b \geq 0$, 
bounding $\Indic*{\{G_n^{k_n} = G_*^{k_n}\}}$ by~1 and taking expectations,
we get
\[
  \ExpecBrackets{B_2(G_n)^m} \;\leq\; d_{\mathrm{TV}}(T_n^{k_n}, T_*^{k_n}) (\log_2 n)^m
    + 2^{m-1} \big(\ExpecBrackets*{\normalsize}{B_2(G_*^{k_n})^m} +
    \ExpecBrackets*{\normalsize}{p_{k_n}^m} (\log_2 n)^m\big).
\]
Recalling that $\log \log n \ll k_n \ll n^{1/2-\epsilon}$ and that
$d_{\mathrm{TV}}(T_n^{k_n}, T_*^{k_n}) = O(k_n / \sqrt{n})$, we see that
the first term on the right-hand side vanishes as $n\to\infty$. So
does $\ExpecBrackets*{\normalsize}{p_{k_n}^m} (\log_2 n)^m$, since
$p_{k_n}^m \leq p_{k_n}$ and, by Lemma~\ref{lemPkSmall},
$\ExpecBrackets{p_{k_n}} = (1-\eta_0)^{k_n}$ with $(1-\eta_0)<1$.
Finally, {$B_2(G_*^{k_n})\leq B_2(G_*)$}. Therefore, to conclude the
proof it suffices to show that $\ExpecBrackets{B_2(G_*)^m}$ is finite.

Let us fix some notation: let $u_k$ denote the {$k$-th} vertex along the
spine of $T_*$, and let $\Head{u_k}$ be the corresponding network in the
blowup construction of $G_*$. Let $\hat{\xi}_k$ be the number of leaves of
$\Head{u_k}$, and for $i = 1, \ldots, \hat{\xi}_k$ let $q_{k, i}$ be
the probability that the directed random walk started from the root of $G_*$
goes through the {$i$-th} leaf of~$\Head{u_k}$. Note that, because of the
leaf-exchangeability of $\Head{u_k}$, we can assume that for all $k\geq 0$
the spine goes through the first leaf of $\Head{u_k}$, so that $q_{k, 1} = p_{k+1}$.
Finally, let us write $G_{k,i}$ for the finite phylogenetic network consisting
of the $i$-th leaf of $\Head{u_k}$ and all of its descendants.  With this
notation, by the grafting property,
\[
  B_2(G_*) \;=\; \sum_{k\geq 0} \Big( p_k\, B_2(\Head{u_k}) +
  \sum_{i=2}^{\hat{\xi}_k} q_{k,i}\, B_2(G_{k,i}) \Big).
\]
Therefore,
\begin{equation} \label{eqEndProofB2Moments}
  \ExpecBrackets{B_2(G_*)^m} \;\leq\; 2^{m-1} \mleft(\!
      \ExpecBrackets*{\bigg}{\!\Big(\sum_{k\geq 0}p_k\,B_2(\Head{u_k})\!\Big)^m}
    + \ExpecBrackets*{\bigg}{\!\Big(\sum_{k\geq 0}\sum_{i=2}^{\hat{\xi}_k} q_{k,i} B_2(G_{k,i})\Big)^m}
  \mright)
\end{equation}
To bound the first term on the right-hand side, note that since
$p_k^m \leq p_{k}$ and since $p_{k}$ is independent of $\Head{u_k}$,
\[
  \ExpecBrackets*{\big}{p_k^m B_2(\Gamma_{u_k})^m} \;\leq\;
  \ExpecBrackets{p_{k}}\, \ExpecBrackets*{\big}{B_2(\Head{u_k})^m}\,.
\]
Using $\ExpecBrackets{p_{k}} = (1-\eta_0)^{k}$ and
$B_2(\Head{u_k}) \leq \log_2 \hat{\xi}_k$, this yields
\[
  \ExpecBrackets*{\big}{p_k^m B_2(\Gamma_{u_k})^m} \;\leq\;
  (1-\eta_0)^{k} \, \ExpecBrackets{(\log_2 \hat{\xi})^m} \,.
\]
Since $\xi$ has a finite variance, it follows that $\hat{\xi}$ has a finite
mean -- and therefore that
$\ExpecBrackets*{\normalsize}{(\log_2 \hat{\xi})^m} < \infty$.
Thus, letting
$\Norm[m]{\cdot} = \ExpecBrackets{\Abs{\,\cdot\,}^m}^{1/m}$ denote the $L^m$-norm,
\[
  \Norm[m]{p_k\,B_2(\Head{u_k})} \;=\; O((1-\eta_0)^{k/m}) \,.
\]
As a result, by the triangle inequality,
\[
  \Big\|\sum_{k\geq 0}p_kB_2(\Head{u_k})\Big\|_m \;\leq\;\;
  \sum_{k\geq 0}\Norm[m]{p_k\, B_2(\Head{u_k})} \;<\; \infty \,.
\]
Finally, to bound the second term in~\eqref{eqEndProofB2Moments}, note that
$\sum_{k\geq 0} \sum_{i = 2}^{\hat{\xi}_k} q_{k, i} = 1$ almost surely.
Thus, Jensen's inequality gives
\[
  \ExpecBrackets*{\bigg}{\Big(\sum_{k\geq 0}\sum_{i=2}^{\hat{\xi}_k}
      q_{k,i}\, B_2(G_{k,i})\Big)^m} \;\leq\;
  \ExpecBrackets*{\bigg}{\,\sum_{k\geq 0}\sum_{i=2}^{\hat{\xi}_k}
    q_{k,i}\, B_2(G_{k,i})^m} \;=\;
  \ExpecBrackets*{\big}{B_2(G)^m},
\]
where $G$ is a blowup of $T$, the Galton--Watson tree with offspring
distribution $\xi$ (the last equality follows from the fact that the $G_{k, i}$'s are
independent of the $q_{k, i}$'s and are all distributed as $G$).
Thus, to finish the proof it suffices to note that, since
$B_2(G) \leq \log_2 \Abs{T}$, where $\Abs{T}$ denotes the number of
leaves of $T$, we have 
\[
  \ExpecBrackets*{\big}{B_2(G)^m} \;\leq\; \ExpecBrackets*{\big}{(\log_2\Abs{T})^m}.
\]
Indeed, it is classic (see also Lemma~\ref{lemEstimatesForLeaves}) that
$\Prob{\Abs{T} \geq n}=\Theta(n^{-1/2})$, which
implies that $(\log_2\Abs{T})^m$ is integrable. This concludes the proof.
\end{proof}

Let us close this appendix by listing further properties of the $B_2$
index of blowups of trees. Although some of these properties are not
related to Galton--Watson trees, the reason why we list them here
is that proving them requires some of the vocabulary and notation of this
Appendix, and that they all follow readily from the following
elementary proposition.

\begin{proposition} \label{propExpectedValueB2Blowup}
Let $\mathbf{t}$ be a fixed tree, and let $G$ be the blowup of $\,\mathbf{t}$
with respect to some family of random networks~$\nu = (\nu_k)_{k\geq 1}$. Let
$X$ denote the directed random walk on $\mathbf{t}$, and for all
$v \in \mathbf{t}$, let $p_v = \Prob{v\in X}$. Then,
\[
  \ExpecBrackets{B_2(G)} \;=\; \sum_{v \in \mathbf{t}} p_v f(d^{+}(v)) \;=\;
  \ExpecBrackets{\sum_{t\geq 0} f(d^{+}(X_t))} \, ,
\]
where $d^+(v)$ denotes the outdegree of $v$ in $\mathbf{t}$, and
$f \colon k \mapsto \ExpecBrackets{B_2(\Head{k})}$, with $\Head{k}\sim \nu_k$
and the convention $f(0) = 0$.
\end{proposition}

\begin{proof}
Let $\rho$ denote the root of $\mathbf{t}$ and, for conciseness, write
$\delta \defas d^+(\rho)$ for its outdegree. Let $\Head{\rho} \sim \nu_{\delta}$
denote the random network associated to $\rho$ in the blowup construction of $G$.
Finally, let $\mathbf{t}_1, \ldots, \mathbf{t}_\delta$ denote the subtrees of
$\mathbf{t}$ subtended by the children of $\rho$, and let
$G_1, \ldots, G_\delta$ be the corresponding blowups with respect to $\nu$.

Since $G$ is obtained by grafting $G_1, \ldots, G_\delta$ on the leaves of
$\Head{\rho}$, by the grafting property we have
\begin{equation} \label{eqProofExpectedValueB2Blowup01}
  B_2(G) \;=\; B_2(\Head{\rho}) \;+\;
  \sum_{i = 1}^{\delta} q_\rho(i) \, B_2(G_i) \,,
\end{equation}
where $q_\rho(i)$ denotes the probability that the directed random walk on
$\Head{\rho}$ ends it its {$i$-th} leaf. Note that $q_\rho$ is a random
probability distribution on $\Set{1, \ldots, \delta}$, and that it follows from
the definition of blowups that:
\begin{mathlist}
  \item being a deterministic function of $\Head{\rho}$,
    $q_\rho$ is independent of $(G_1, \ldots, G_\delta)$;
  \item by the leaf-exchangeability of $\Head{\rho}$, $q_\rho$ is exchangeable.
\end{mathlist}
As a result, taking expectations in \eqref{eqProofExpectedValueB2Blowup01}
we get
\[
  \ExpecBrackets{B_2(G)} \;=\; f(\delta) \;+\;
  \frac{1}{\delta} \sum_{i = 1}^{\delta} \ExpecBrackets{B_2(G_i)} \,,
\]
and $\ExpecBrackets{B_2(G)} = \sum_{v \in \mathbf{t}} p_v f(d^{+}(v))$ follows by
induction. Finally, to see that this is also
$\ExpecBrackets{\sum_{t\geq 0} f(d^{+}(X_t))}$, it suffices to note that
$p_v = \sum_{t\geq 0} \Prob{X_t = v}$ for any internal vertex $v$, and that
$\sum_{v \in \mathbf{t}} \Indic{X_t = v} = 1$ a.s.\ for all $t\geq 0$.
\end{proof}

By ``not blowing up'' the base tree $\mathbf{t}$ in
Proposition~\ref{propExpectedValueB2Blowup} (i.e.\ by using star trees for the
networks by which internal vertices are replaced in the blowup, so that the
tree is left unchanged), we immediately get the following simple expression for
the $B_2$ index of a tree. This expression does not seem to have
been pointed out previously in the literature.

\begin{corollary} \label{corAlternativeFormulaB2Tree}
For any tree $\mathbf{t}$, we have
\[
  B_2(\mathbf{t}) \;=\;
  \sum_{v \in \mathbf{t}} p_v \log_2 d^+(v) \,, 
\]
where $d^+(v)$ denotes the outdegree of $v$ in $\mathbf{t}$, and
$p_v = \big(\prod_u d^{+}(u)\big)^{-1}$, where the product runs over the
vertices on the path from the root of $\mathbf{t}$ to $v$, excluding $v$.
\end{corollary}

Because $f(k) \leq \log_2 k$, Proposition~\ref{propExpectedValueB2Blowup}
also has the following corollary.

\begin{corollary} \label{corIneqB2Blowup}
Let $\mathbf{t}$ be a fixed tree, and let $G$ be a blowup of $\,\mathbf{t}$.
Then,
\[
  \ExpecBrackets{B_2(G)} \; \leq\; B_2(\mathbf{t}) \,.
\]
\end{corollary}

Note however that the inequality in Corollary~\ref{corIneqB2Blowup} only holds
after integrating with respect to the blowup procedure. In particular,
if $T$ is a random tree and $G$ is a blowup of $T$, then $B_2(T)$ does not
necessarily stochastically dominate $B_2(G)$. To see this, let $T = \mathbf{t}$
be deterministic, and note that it is then possible to have
$B_2(G) > B_2(\mathbf{t})$. This implies that, in that case, there is not
monotone coupling of $B_2(G)$ and $B_2(T)$.

Nevertheless, as the next corollary shows, $B_2(T)$ is second-order
stochastically dominant over $B_2(G)$. Very loosely speaking, this means that --
in addition to being no smaller in expectation -- $B_2(T)$ is more predictable
than $B_2(G)$.

\begin{corollary}\label{corStochasticDominance}
Let $T$ be a tree, and let $G$ be a blowup of $T$.
Then, $B_2(T)$ is second-order stochastically dominant over $B_2(G)$, that is:
for every nondecreasing concave function $\phi$,
\[
  \ExpecBrackets{\phi(B_2(G))} \; \leq\; \ExpecBrackets{\phi(B_2(T))} \,.
\]
\end{corollary}

\begin{proof}
Let $\phi$ be a nondecreasing concave function. First, by
Corollary~\ref{corIneqB2Blowup}, we have $\ExpecBrackets{B_2(G) \given T} \leq B_2(T)$,
and therefore
\[
  \phi\big(\ExpecBrackets{B_2(G) \given T}\big) \;\leq\; \phi(B_2(T)) \,.
\]
As a result, by Jensen's inequality,
\[
  \ExpecBrackets*{\big}{\phi(B_2(G))} \;=\;
  \ExpecBrackets*{\big}{\ExpecBrackets{\phi(B_2(G)) \given T}} \;\leq\;
  \ExpecBrackets*{\big}{\phi\big(\ExpecBrackets{B_2(G) \given T}\big)} \;\leq\;
  \ExpecBrackets*{\big}{\phi(B_2(T))}\,,
\]
concluding the proof.
\end{proof}

Finally, note that Proposition~\ref{propExpectedValueB2Blowup} yields a simple derivation
of the expression that was already given in Theorem~\ref{thmExpecB2BlowUpGW}
for the expected value of the $B_2$ index of blowups of Galton--Watson trees.
Before detailing this, let us briefly compare the two approaches:
recall that the proof of Theorem~\ref{thmExpecB2BlowUpGW} used the recursive
structure of Galton--Watson trees (i.e.\ the branching property) together with
the independence between the base tree and the networks used in the blowup
procedure to get a distributional equation for $B_2(G)$. In a way,
Proposition~\ref{propExpectedValueB2Blowup} isolates the part of the
proof of Theorem~\ref{thmExpecB2BlowUpGW} where the independence between
the base tree and the networks is used, and in the proof below
we use the branching property to recover the expression of
Theorem~\ref{thmExpecB2BlowUpGW}; what justifies presenting
the two proofs is that the way the branching
property is used is different.

\begin{corollary} \label{corExpectedValueB2GaltonWatson}
Let $T$ be a Galton--Watson tree with offspring distribution $\xi$
and let $G$ be a blowup of $T$ with respect to a
family of networks $\nu = (\nu_k)_{k\geq 1}$. Then, 
\[
  \ExpecBrackets{B_2(G)} \;=\; \frac{\ExpecBrackets{f(\xi)}}{\Prob{\xi = 0}} \,, 
\]
where $f(0) = 0$ and, for $k\geq 1$, $f(k) = \ExpecBrackets{B_2(\Head{k})}$,
where $\Head{k} \sim \nu_k$, and with the convention $0/0 = 0$.
\end{corollary}

\begin{proof}
First, let us deal with degenerate cases: when $\Prob{\xi = 1} = 1$,
the tree $T$ is an infinite path, and therefore $G$ has exactly one end --
which implies $B_2(G) = 0$. Since $f(1) = 0$, the proposition holds with the
convention $0/0 = 0$. In the case where $\xi$ is not almost surely equal to~1
and where $\Prob{\xi = 0} = 0$, we have $B_2(G) = +\infty$ and the proposition
also holds.

Now, assume that $\Prob{\xi =0}> 0$. Note that, by the branching property,
\[
  \sum_{t\geq 0} f(d^{+}(X_t)) \;\mathrel{\overset{d}{=}}\;
  \sum_{i = 1}^{N} f(\xi^+_i)
\]
$(\xi_i^+)_{i\geq 1}$ are i.i.d.\ copies of
$(\xi \mid \xi > 0)$, the variable $\xi$ conditioned to be positive, and
where $N$ is a geometric variable on $\Set{0, 1, \ldots}$
with parameter $\eta_0 \defas \Prob{\xi = 0}$ that is independent of
$(\xi_i^+)_{i\geq 1}$. As a result, by Wald's formula,
\[
  \ExpecBrackets{B_2(T)} \;=\; \ExpecBrackets{N}\, \ExpecBrackets*{\normalsize}{f(\xi^+)} \;=\;
  \frac{1 - \eta_0}{\eta_0}\cdot \frac{\ExpecBrackets*{\normalsize}{f(\xi) \Indic{\xi > 0}}}{1 - \eta_0} 
  \;=\; \frac{\ExpecBrackets{f(\xi)}}{\eta_0}\,, 
\]
because $f(0) = 0$. This concludes the proof.
\end{proof}

\end{document}